\documentclass[10pt,]{article}
\usepackage{amsfonts,bbm}
\usepackage{epsfig,amsmath,graphics,mathabx,graphicx,mathrsfs,amsthm}
\usepackage{eepic,bm,epsfig,color,subfigure}
\DeclareMathAlphabet{\mathpzc}{OT1}{pzc}{m}{it}

\def\cstar{C^{*}}
\def\comp{\mathbb{C}}

\newcommand{\beq}{\begin{equation}}
\newcommand{\eeq}{\end{equation}}
\newcommand {\cali}[1]{{\mathcal #1}}

\newcommand{\tr}{{\tt Tr}} 
 
\newcommand{\conj}[1]{\overline{#1}}
 
\newcommand{\pmat}{\begin{pmatrix}} 
\newcommand{\emat}{\end{pmatrix}} 
\newcommand{\complex}{\mathbb{C}}
\newcommand{\real}{\mathbb{R}}
\newcommand{\norm}[1]{|\!|#1|\!|} 
\newcommand{\tensor}{\otimes}
\newcommand{\unit}{\mathbbm{1}}

\newcommand{\Sp}{{\tt sp}}
\newtheorem{propn}{Proposition}{}{}
\newtheorem{thm}{Theorem}{}{}
\newtheorem{lem}{Lemma}{}{}
\newtheorem{cor}{Corollary}{}{}
{}{}
\newtheorem{cor_t}{Corollary}[thm]{}{}
\newtheorem{Def}{Definition}{}
\newcommand{\commentout}[1]{}
\newcommand{\be}{\begin{enumerate}}
\newcommand{\ee}{\end{enumerate}}
\newcommand{\bi}{\begin{itemize}}
\newcommand{\ei}{\end{itemize}}
\newcommand{\floor}[1]{\lfloor #1\rfloor}
\begin{document}
\title{An algebraic framework for information theory: Classical Information} 
\author{Manas K. Patra and Samuel L. Braunstein \footnote{{\bf email}:\{manas,schmuel\}@cs.york.ac.uk}\\
Department of Computer Science \\ University of York, York YO10 5DD, UK}
\date{}
\maketitle
\begin{abstract}
This work proposes a complete algebraic model for classical information theory. As a precursor the essential probabilistic concepts have been defined and analyzed in the algebraic setting. Examples from probability and information theory demonstrate that in addition to theoretical insights provided by the algebraic model one obtains new computational and anlytical tools. Several important theorems of classical probahility and information theory are formulated and proved in the algebraic framework.
\end{abstract}
\section{Introduction}
The present paper proposes an algebraic model of classical information theory. We then carry out a detailed investigation of the model. The connection between operator algebras and information theory---both classical and quantum---have appeared in the scientific literature since the beginnings of information theory and operator algebras---both classical and quantum (see e.g.\ \cite{Umegaki4, Segal60,Lindblad,araki75,Keyl,Beny,Kretschmann}). The standard formulation of classical information theory \cite{Ash, CoverT} on the other hand is sometimes seen as an important application of probability theory. Thus probabilistic concepts like distribution function, conditional expectation and independence are vital for the development of information theory. Most previous work including those mentioned above focus on some aspects of information theory, especially the noncommutative generalizations of the concepts of entropy and for specific probabilistic concepts they often resort to a representation on some Hilbert space. As a consequence, there does not appear to be a unified coherent approach based on intrinsically algebraic notions. The construction of such a model is one of the goals of the paper. As probabilistic concepts play such an important role in the development of information theory we devote a fairly large section to  an algebraic approach to probability. It was I. E. Segal \cite{Segal54}, one of the major players in the early development of operator theory who first proposed such an algebraic approach to probability theory. Although we have mostly restricted ourselves to the discrete case, sufficient for our models of communication and information processes, our proposed model is different from Segal's. We believe several aspects of our approach are novel (see the section-wise synopsis below) and yield deeper insights to information processes. 
  
A strong motivation for this paper is the relatively young field of quantum information theory. It is almost folklore that in quantum mechanics we are forced to deal with noncommutative entities. Thus, the language of $\cstar$ algebras, already known to physicists for decades \cite{Haag,Emch} as ``the algebra of observables''  on which many extensions of classical probabilistic concepts can be made, became a natural setting for quantum information. As a complex quantum information scheme or protocol has several {\em classical} components (e.g. classical communication, coin-tosses etc.) it is important that we have a unified model and a single language for quantum and classical information. Such a formulation will be of great help in the difficult task of protocol analysis. Besides a unified framework will be of significant advantage for theoretical analysis. For example, a deeper study of quantum phenomena like (no) quantum broadcasting \cite{Barnum}, quantum Huffman coding \cite{Braunstein}, channel capacity \cite{Schumacher} to name a few would benefit from the investigations of these structures. 
In this framework we may view a classical process as a special type of process described by {\em commuting} elements. Therefore, it seems appropriate to investigate this special case first. As we will show the classical structure is quite rich and sheds new light on some familiar aspects of information theory. There is yet another reason. In quantum mechanics we have several examples of observables taking only a finite number of values (the spectrum is finite). But in classical mechanics all variables take on a continuum of values. Therefore, we often see statements like ``a finite-dimensional operator like spin is a purely quantum phenomenon that has no classical analogue''. However, when we talk about information systems finite-dimensional quantum systems have obvious classical analogues. A 2-dimensional quantum ``source'' corresponds to a classical binary source. Our investigations raise some questions about the possibility of an alternative formulation of probability theory with a more algebraic flavour \cite{Segal54}. This is interesting in itself. But it is a side issue in this paper and will only be briefly commented upon. Since our main concern is the mathematical models of information processing systems we will be primarily dealing with discrete systems, thus circumventing some tricky topological  issues. 

Let us recall a simple model of a communication system proposed by Shanon \cite{Shannon48,ShannonWeaver}. 
\commentout{
\begin{figure}[h] \label{fig1}
\centering
\includegraphics[height=2.5cm,width=1.0\textwidth]{commModel.jpg}
\caption{A simple model of communication system}
\end{figure}
}
This model has essentially four components: source, channel, encoder/decoder and receiver. The source could be representing very different kinds of objects: a human speaker, a radar antenna or a distant star. We usually have some model of the source. The coding/decoding operation is required for three basic reasons: i) the source/receiver alphabet and the channel alphabet may be different, ii) to maximize the rate of information communication and ii) to detect and correct errors due to noise and distortion. Some amount of noise affects every stage of the operation. So the behavior of components are generally modeled as stochastic processes. This is valid in both the classical and often quantum communication processes. The difference, of course, is in the description of the two processes. As in any stochastic process, we specify the source by a family $X_t$ of {\em random variables} and the various stages of the communication system are modeled as (stochastic) transformations of these variables. The parameter $t$ can be continuous or discrete. In this work our primary focus will be on discrete processes corresponding to discrete time. Thus, a discrete source can be viewed as a generator of a countable set of random variables. Let us suppose the source ``tosses'' a coin and sends a 1 if it is ``heads'' and a 0 otherwise. We may model this by a pair of random variables $\{X_h,X_T\}$ on the probability space $\{h\text{ (heads)},t\text{ (tails)}\}$ such that $X_h(h)=X_t(t)=1 \text{ and }X_h(t)=X_t(h)=0$. If the coin is unbiased we say that the {\em state} of the source is given by a probability measure $\{1/2,1/2\}$. In general it is $\{p,q\}$ where $0\leq p=1-q\leq 1$ is the probability of heads. This simple model can be generalized to more complicated sources. Besides these elementary random variables we encounter functions of these variables. Thus we are led to study {\em algebras} of random variables. Recall the standard definition of a random variable: it is a (measurable)  function on a probability space $S$. Hence, in the standard formulation we need a probability space or sample space to define our random variables or ``observables''. Let us also recall that a probability space is a {\em triple} $(S, \cali{M}, \mu)$, where $S$ is a set, the set of elementary or atomic events, $\cali{M}$ is a $\sigma$-algebra of subsets and $\mu$ is the probability measure. Thus if $\{A_n\}$ is a sequence of mutually disjoint elements from $\cali{M}$ then
\[ \mu(\bigcup_n A_n)=\sum_{n=1}^{\infty} \mu(A_n). \]
Moreover, $\mu(S)=1$ and $\mu(B)\geq 0$ for any $B\in \cali{M}$. These are essentially the Kolmogorov axioms. A real or complex valued random variable is a measurable function form $S\rightarrow \real$ or 
$S\rightarrow \complex$. Here measurability is with respect to the Borel $\sigma$-algebra of $\real$ or $\complex$. We recall that the Borel $\sigma$-algebra of any 
topological space is generated by its open sets. So in some sense in this formulation the probability space is fundamental and the notion of random variables is based on the former. However, from an observer/experimenter point of view the random variables are the basic entities because these are precisely the observables. In statistical theories like information theory it is the set of random variables and their distributions and transformations which are of primary interest. Of course, to compute the probability distributions of the random variables we have to appeal to the original probability space. But once the distributions have been determined for almost all computations they suffice and the underlying probability or sample space plays little role. The fundamental theorem of Kolmogorov \cite{Billingsley, Shiryayev} guarantees that given a set of random variables and their distributions satisfying certain consistency conditions we can reconstruct a probability space giving these distributions. These observations suggest that we take the algebra of random variables or observables as our primary structure and derive all relevant quantities from this structure. One of the advantages is that we deal with a smaller spaces restricted to quantities of interest. In the modeling of security protocols this a more realistic approach since different participants have access to different sets of observables and may assign different probability structures on the same set of events. 

In the quantum case there are more fundamental reasons for working with the algebras of observables. We will not go into these here. The current work is an attempt at formulating (classical) information theory in an algebraic framework. We will mainly focus on $\cstar$ and von Neumann algebras. We will see that most interesting spaces of observables do have a $\cstar$ structure. As mentioned before, we will be dealing with discrete spaces in this work. We also observe that $\cstar$ algebras have been studied intensively since the pioneering works of Murray, von Neumann, Gelfand, Naimark and Segal and others starting from 1930s. As we stated at the beginning of this section, several probabilistic and information theoretic concepts like conditional expectation, entropy, differential entropy have previously been investigated in the algebraic context. However to the best of our knowledge there is no work investigating information and communication theory in a purely algebraic framework. Our investigations indicate that most if not all important concepts and constructs of information theory can be dealt with in the algebraic framework. The paper is structured as follows. 

In Section \ref{sec:algebra} we give the basic definitions of the algebras of interest. This section is fairly detailed as we provide proofs of several structure  theorems for finite-dimensional abelian $\cstar$ algebras and their tensor products, possibly infinite. There are two reasons for this. The first is to make the paper as self-contained as possible. The second reason is to demonstrate the power and utility of the algebraic techniques. Moreover, we believe that in these special cases some of the proofs are new. We also give several examples. 

Section \ref{sec:prob} gives an account of probabilistic concepts from an algebraic perspective. In particular, we investigate the fundamental notion of independence and demonstrate how it relates to the algebraic structure. We note that there is a very sophisticated theory of noncommutative or ``free probability''\cite{Voic}. Our approach in the simpler commutative case is different in several aspects. One important point in which our approach seems novel is the definition of a probability distribution function. The definition we give is algebraic in the sense that it depends on the intrinsic properties of the algebra.  Specifically, we define a probability distribution function as the weak limit of a net or sequence of elements in a subalgebra representing an approximate identity of an ideal or a subalgebra. To illustrate the practical use of these techniques we give some typical examples from standard probability theory. The problem of ``waiting time'' shows that the algebraic approach can offer new techniques and insights. Finally, using the definition of distribution function and some other constructs we formulate and prove some of the basic limit theorems in this framework. These are used later in proving results in information theory. 

In Section \ref{sec:info} we give a precise algebraic model of information communication system. The fundamental concept of entropy is introduced as a limiting value of typical sequences of the algebra. The notion of typical sequence comes from the limit theorems. In the conventional approach the limit is taken in the probability (convergence in measure). In our algebraic case it corresponds to a weak limit. The point is, we can do all this in purely algebraic setting. We also define and study the crucial notion of a channel. In particular, the {\em channel coding theorem} is presented as an approximation result. Stated informally, 

{\em Every channel other than the useless ones can be approximated by a lossless channel under appropriate coding.} 

In the final section we summarize our constructions and discuss future work. 

\section{Algebraic Preliminaries} \label{sec:algebra}
An algebra $\cali{A}$ is vector space over a field $\tt{F}$ with an associative bilinear product: $\cali{A}\times \cali{A}\rightarrow \cali{A}$. 
We take ${\tt F}=\comp$, the field of complex numbers. We deal mostly with {\em unital} algebras, that is, algebras with a unit 
$\unit$. A Banach algebra is an algebra with a non-negative real function $\norm{}$ on $\cali{A}$ such that
\[
\begin{split}
\norm{x} & \geq 0 \text{ and  $\norm{a}=0$, iff $a=0$} \\
\norm{x+y} &\leq \norm{x}+\norm{y} \text{ (triangle inequality)}\\
\norm{xy} & \leq \norm{x}\norm{y} \text{ (Banach property) } \\ 
\end{split}
\]
and $\cali{A}$ is {\em complete} in the topology defined by the norm. A $\cstar$ algebra $B$ is a Banach algebra with 
an anti-linear involution $^*$ (a map $\sigma$ is an involution if $\sigma^2=1$, it is antilinear if $\sigma(x+c y)=\sigma(x)+\conj{c}\sigma(y),\; c \text{ a complex number}$) such that 
\[\norm{xx^*}=\norm{x}^2\text{ and } (xy)^*=y^*x^*\forall x,y \in B\] 
This implies that $\norm{x}=\norm{x^*}$. The quintessential examples of a $\cstar$ algebra are the norm-closed subalgebras of $\cali{L}(H)$, 
the set of bounded operators on a Hilbert space of $H$. The fundamental Gelfand-Naimark-Segal ({\bf GNS}) theorem states that 
every $\cstar$ algebra can be isometrically embedded in some $\cali{L}(H)$. The notion of the spectrum  of an operator 
has an algebraic analogue without reference to the representation space. The resolvent of an element $x$ in the $\cstar$ algebra $B$ is 
the set $R(x)\subset \comp$ such that $\lambda \in R(x)$ implies $\lambda\unit -x$ is invertible. The spectrum $\Sp(x)$ is the complement of the resolvent. 
The spectrum is a nonempty closed and bounded subset and hence compact. Define $r(x)=\sup\{|\lambda|: \lambda \in \Sp(x)\}$, the spectral radius. A 
basic result states that 
\[ r(x)=\lim_{n\rightarrow \infty} \norm{x^n}^{1/n}.\]
An element $x$ is self-adjoint if $x=x^*$, normal if $x^*x=xx^*$ and positive (strictly positive) if $x$ is self-adjoint and $\Sp(x)\subset [0,\infty)((0,\infty))$. A self-adjoint
element has a real spectrum and conversely. Since $x=(x+x^*)/2+i(x-x^*)/2i$ any element of a $\cstar$ algebra can be decomposed into self-adjoint ``real'' ($(x+x^*)/2$) and ``imaginary''  ($(x-x^*)/2i$) parts. For a self-adjoint element $x$, $r(x)=\norm{x}$. 
Thus a positive element is self-adjoint. The positive elements define a partial order on $A$. Thus 
$x\leq y$ iff $y-x\geq 0$ (positive). An important property of positive elements is that they have unique square-roots. Thus if $a\geq 0$ there is 
a unique element $b \geq 0$ such that $b^2=a$. We write $\sqrt{a}\text{ or } a^{1/2}$ for the square-root. Since $x^*x\geq 0$ it has a unique square-root. If $x$ is normal we write $|x|=\sqrt{x^*x}$. In particular, if $x$ is self-adjoint, $|x|=\sqrt{x^2}$.  A self-adjoint element $x$ has a decomposition $x=x_+-x_-$ into positive and negative parts where \(x_+=(|x|+x)/2\ \text{ and } x_-=(|x|-x)/2)\) are positive. An element $p\in B$ is a projection 
if $p$ is self-adjoint and $p^2=p$. Given two $\cstar$-algebras $A$ and $B$ a homomorphism $F$ is a linear map preserving the product 
and $^*$ structures. It is continuous iff bounded. A continuous isomorphism of $\cstar$ algebras is an isometry (norm preserving). 
A homomorphism is positive if it maps positive elements to positive elements. A (linear) functional on $A$ is a linear map $A\rightarrow \comp$. 
The GNS construction starts with a  positive  functional (mapping positive elements to non-negative numbers) on $B$. The details may be found in \cite{KR1,Tak1}. A positive functional $\omega$ such that $\omega(\unit)=1$ is called a {\em state}. The set of states $G$ is convex. The extreme points 
are called {\em pure states} and $G$ is the convex closure of pure states (Krein-Millman theorem). A set $B\subset A$ is called a subalgebra if it is a 
$\cstar$ algebra with the inherited product. That is, it is a subalgebra in the algebraic sense and it is {\em closed} in the norm topology. 
A subalgebra is $B$ called unital if it contains the identity of $A$.
Our primary interest will be on {\em abelian} (also called commutative) algebras. The structure theory 
is a bit different in this case. Of course, the GNS construction is valid and the elements of the algebra act as multiplication 
operators on the representing Hilbert space. However, there is an alternative representation in the abelian case 
due to Gelfand and Naimark which will be of primary interest to us. To motivate it consider an example. 

Let $X$ be a compact Hausdorff topological space, for example, a closed and bounded set in $\real^n$. Let $C(X)$ denote the space continuous 
complex functions on $X$. It includes the constant functions. If we define addition and multiplication point-wise
\beq \label{eq:exampBasic1}
\begin{split}
(f+g)(x)&=f(x)+g(x),\; (fg)(x)=f(x)g(x)\text{ and } \\
\norm{f}&=\sup_{x\in X} |f(x)|\; \forall f,g\in C(X)
\end{split}
\eeq 
then $C(X)$ becomes a complex Banach algebra. If we define $f^*(x)=\conj{f(x)}$ then $C(X)$ is an abelian $\cstar$ algebra. This is a prototype of abelian $\cstar$ algebras \cite{KR1}. One can generalize to (essentially) bounded measurable functions on measure spaces with appropriate norm. However, for the purposes of this paper it suffices to consider compact spaces with measures defined on Borel $\sigma$-algebras. We will dwell more on this point in the next section. A complex function (not necessarily continuous) is called simple if its range is finite. For example, the {\em indicator} function $I_S$ of a subset $S\subset X$, given by $I_S(x)=1 \text{ if } x\in S$ and 0 otherwise is a simple function. It is not continuous unless $S=X$ or $S$ is a connected component. Simple functions play a crucial role in probability and integration theory. From their definition it follows that the projections in $C(X)$ are precisely the indicator functions. The constant functions 1 and 0 are both projections corresponding to $S=X\text{ and } \emptyset$ resp. These are the only projections in $C(X)$ if $X$ is connected. The basic structure theorem for abelian $\cstar$ algebras is the following. 
\begin{thm}
An abelian $\cstar$ algebra with unity is isomorphic to the algebra $C(X)$ for a compact Hasudorff space $X$. The isomorphism is an isometry (norm preserving). 
\end{thm}
The main idea of the proof comes from the following observation. In any function algebra $C(X)$ for $p\in X$ the map
$\sigma_P:f\rightarrow f(p),\; f\in C(X)$ is a linear functional on $C(X)$. These are multiplicative functionals in the sense that $\sigma_p(xy)= \sigma_p(x)\sigma_p(y)$. In fact these are the only possible multiplicative functionals. The Gelfand representation for an abstract abelian $\cstar$ algebra $A$ identifies the space $X$ as the set of multiplicative functionals and gives it a topology to make these continuous. The details can be found in \cite{KR1}. 

Now let $X=\{a_1, \dotsc, a_n\}$  be a finite set with discreet topology. Then $A=C(X)$ is the set of all functions $X\rightarrow \comp$. The algebra $C(X)$ can be considered as the algebra of (complex) random variables on the finite probability space $X$. Let $x_i(a_j)=\delta_{ij},\; i,j=1,\dotsc, n$. Here $\delta_{ij}=1 \text{ if }i=j \text{ and } 0$ otherwise. The functions $x_i\in A$ form a basis for $A$. Their multiplication table is particularly simple: $x_ix_j=\delta_{ij}x_i$. They also satisfy $\sum_i x_i=\unit$. These are projections in $A$. They are orthogonal in the sense that $x_ix_i=0\text{ for }i\neq j$. We call any basis consisting of elements of norm 1 with distinct elements orthogonal {\em atomic}. A set of linearly independent elements $\{y_i\}$ satisfying $\sum_i y_i=\unit$ is said to be complete. The next theorem gives us the general structure of any finite-dimensional algebra.  
\begin{thm}\label{thm:structFinite}
Let $A$ be a finite-dimensional abelian $\cstar$ algebra. Then there is a unique (up to permutations) complete atomic basis $\cali{B}=\{x_1, \dotsc, x_n\}$. That is, the basis elements satisfy
\beq \label{eq:structFinite}
x_i^*=x_i,\; x_ix_j=\delta_{ij}x_i,\;  \norm{x_i}=1 \text{ and }\sum_i x_i =\unit,\;
\eeq
Let $x=\sum_i a_ix_i\in A$. Then $\Sp(x)=\{a_i\}$ and hence $\norm{x}=\max_i\{|a_i|\}$. 
\end{thm}
\begin{proof}
Let $\{y_1,\dotsc,y_n\}$ be a basis for $A$. Since the self-adjoint elements $(y_i+y_i^*)/2$ and $i(y_i-y_i^*)/2$ span $A$ we can choose an independent set. Hence, we may assume that the $y_i$ are self-adjoint. Then each $y_i^2$ is positive and hence possesses a square-root $|y_i|$. Moreover, $|y_i|\geq y_i$ \cite{KR1,Bratteli1}.\footnote{Bratteli (pp. 35) gives a proof which does not use Gelfand representation.} We can therefore write each 
$y_i=(|y_i|+y_i)/2 - (|y_i|-y_i)/2$, as the difference of two positive elements. Again, choosing an independent set we may assume that $y_i$ themselves are positive with norm 1. Let $S=\{z:z\geq 0 \text{ and } \norm{z}\leq 4$. $S$ is convex and compact (being closed and bounded) and $y_i\in S$. Hence, by the Krein-Millman theorem \cite{KR1} $S$ is the convex closure of its {\em extreme points}\footnote{Recall that extreme points of a convex set are those which cannot be written as a non-trivial convex combination of some members of the set}. We may assume that these extreme points have norm 1 (obviously discarding 0). Since each $y_i$ can be written as a finite convex sum of its extreme points we can pick a basis $x_1,\dotsc,x_n$ of extreme points. We complete the proof by showing that the $x_i$'s satisfy equations (\ref{eq:structFinite}) and that they are unique. 

Now $\norm{x_i}=1$ implies that for any $|\lambda| >1$, $\lambda-x_i=\lambda(\unit-\lambda^{-1}x_i)$ is invertible. This can be proved by using the geometric series of $(1-\lambda^{-1}x_i)^{-1}$. Hence if $a\in {\tt sp}(x_i)$ then $0\leq a \leq 1$ and $\unit-x_i$ is positive. Since ${\tt sp}(x_i-x_i^2)=\{a-a^2:a\in{\tt sp}(x_i)\}$ and $a-a^2\geq 0$ it follows that $x_i-x_i^2\geq 0$. As $x_i=(2(x_i-x_i^2)+2x_i^2)/2$, a convex combination of two positive elements in $S$ and $x_i$ is a non-zero extreme point we must have $x_i-x_i^2=0$ or $x_i-x_i^2=x_i^2$. The last possibility is ruled out because it would imply $\norm{x_i}=2\norm{x_i^2}=2\norm{x_i}^2=2$. Hence $x_i^2=x_i$. To prove that they are orthogonal observe that 
$x_i -x_ix_j=x_i(1-x_j)$ is positive. Thus $x_i=(2x_i(1-x_j)+2x_ix_j)/2$ is a convex combination of points in $S$. Hence, as before either $x_ix_j=0$ or $x_i=x_ix_j$. With $x_j$ in place of $x_i$ we conclude that  $x_ix_j=0$ or $x_j=x_ix_j$. Thus the only possibility for $x_i\neq x_j$ is that $x_ix_j=0$. 

To prove the decomposition property let $\unit=\sum_ia_ix_i$. Squaring and using the orthogonality of $x_i$'s we conclude that $a_i=1 \text{ or }0$. If some $a_k=0$ then $x_k=x_k\unit=x_k\sum_ia_ix_i=0$. Hence, all $a_k=1$. Finally, let $\{z_i\}$ be another basis satisfying (\ref{eq:structFinite}). Let $z_i=\sum_j b_{ij}x_j$. As before, $b_{ij}=1\text{ or } 0$ and the matrix $(b_{ij})$ is a 0-1 matrix. For fixed $i$ let $T_i$ be the set of integers $j$ such that $b_{ij}=1$. Then $x_ix_j=0,\; i\neq j$ implies $T_i$ and $T_j$ are disjoint. This along with the last condition in (\ref{eq:structFinite}) implies that $T_i$'s form a partition of the set $\{1,\dotsc, n\}$. Thus each $T_i$ is a singleton and the matrix $(b_{ij})$ has exactly one 1 in each row and column. It is a permutation matrix. 

Let $x=\sum_i a_ix_i$ be an element of $A$. Then $\lambda\unit -x=\sum_i(\lambda-a_i)x_i$. This is invertible iff $\lambda\neq a_i,\;i=1,\dotsc,n$ with inverse $\sum_i(\lambda-a_i)^{-1}x_i$. The proof is complete. 
\end{proof}
Let us observe that we could have proved the theorem using the Gelfand representation. But the above proof is more intrinsic depending mostly on the structure of the algebra itself only. 
\begin{cor}
Let $A$ be an abelian $\cstar$-algebra satisfying the following conditions. There are finite-dimensional subalgebras $A_k,\; k=0,1,\dotsc$ with 
\[ 
A=\bigcup_{k=0}^{\infty} A_k \text{ and }A_k\subset A_{k+1}\; \forall k
\]
and for each $k$ corresponding to $A_k$ there is complementary subalgebra $A_k'\subset A_{k+1}$ such that $A_kA_k'=A_{k+1}$, $A_{k}\bigcap A_k'=\{0, \unit\}$ and for $x\in A_k,y\in A_k'$ implies $xy\neq 0$ unless $x$ or $y$ is 0. Then there is a countable basis for $A$ satisfying the first three equations in (\ref{eq:structFinite}). 
\end{cor}
\begin{proof}
We prove by induction. The case of $A_0$ is proved in the theorem. Assume we have an atomic basis $\{y^n_1, \dotsc, y^n_{k_n}\}$ for $A_n$. There is a (unique) atomic basis $\{x^n_{1}\,\dotsc, x^n_{m_n}\}$ in $A_n'$. It is now a routine matter to show that $\{x^n_iy^n_j: 1\leq k_n \text{ and } 1\leq m_n\}$ form a basis in $A_{n+1}$. 
\end{proof}
The conditions in the corollary can be slightly weakened by requiring that there be embeddings (injective algebra homomorphisms) $\alpha_k:A_k\rightarrow A_{k+1}$ and  $\alpha'_k:A'_k\rightarrow A_{k+1}$ such that the images $\alpha_k(A_k)$ and $\alpha'_k(A'_k)$ satisfy the conditions stated. Such a structure will appear in the {\em tensor product} of algebras to be defined below. They play an important role in our modeling of information and communication systems. Let us also note that the basis structure in Theorem \ref{thm:structFinite} may be used to defined a finite dimensional $\cstar$ algebra abstractly.
\subsection{Tensor products}
We next describe an important construction for $\cstar$ algebras. Given two $\cstar$ algebras $A$ and $B$, the tensor product $A\tensor B$ is defined as follows. As a set it consists of all finite linear combinations of symbols of the form $\{x\tensor y:x\in A,y\in B\}$ subject to the conditions that for all \(x,u\in A, \; y,z\in B \text{ and } c\in\comp\), 
\beq
\begin{split}
%\begin{eqnarray}
(cx)\tensor y &= x\tensor (cy)= c(x\tensor y)  \\ 
(x+u)\tensor y &= x\tensor u+u\tensor y\text{ and } x\tensor (y+z)=x\tensor y+x\tensor z.
%\end{eqnarray}
\end{split}
\eeq
Thus the tensor product is {\em bilinear}. There are no other relations. Note that by definition the products of the form $x\tensor y$ span $A\tensor B$. Hence, if $\{x_i\} \text{ and } \{y_j\}$ are bases for $A$ and $B$  respectively then $\{x_i\tensor y_j\}$ is a basis for $A\tensor B$. The linear space $A\tensor B$ becomes an algebra by defining 
$(x\tensor y)(u\tensor z)=xu\tensor yz$ and extending by bilinearity. Explicitly, 
\[\sum_i a_i(x_i\tensor y_i)\sum_j b_j (u_j\tensor z_j)=\sum_{ij}a_ib_j(x_iu_j\tensor y_iz_j)\]
The $*$ is defined by $(x\tensor y)^*=x^*\tensor y^*$ and extending {\em anti-linearly}. The problem is defining the norm since it is not a linear function. In fact, for general $\cstar$ 
algebras there could be a number of inequivalent norms on different completions of $A\tensor B$. This problem of non-uniqueness, however, does not exist if one of the factors is abelian or finite-dimensional. Since, in this work we will be primarily concerned with abelian algebras this point will not be discussed further. Our basic model will be an {\em infinite} tensor product of finite dimensional $\cstar$ algebras which we present next. 

\newcommand{\inftens}[1]{\bigotimes^{\infty}{#1}}
Let $A_k,\;k=1,2,\dotsc,$ be finite dimensional abelian $\cstar$ algebras with atomic basis $B_k=\{x_{k1},\dotsc,x_{kn_k}\}$. Let $B^{\infty}$ be the set consisting of all infinite strings of the form \(z_{i_1}\tensor z_{i_2}\tensor \cdots\) where all but a finite number ($>0$) of $z_{i_k}$s are equal to $\unit$ and if some $z_{i_k}\neq \unit$ then $z_{i_k}\in B_k$. Explicitly, $B^{\infty}$ consists of strings of the form $z_{i_1}\tensor z_{i_2}\tensor \cdots \tensor z_{i_k}\tensor\unit\tensor\unit\tensor\cdots,\;k=1,2,\dotsc $ and $z_i\in B$. Let $\tilde{\mathfrak{A}}=\tensor_{i=1}^{\infty}{A}_i $ be the vector space with basis $B^{\infty}$ such that $z_{i_1}\tensor z_{i_2}\tensor \cdots \tensor z_{i_k}\tensor\cdots $ is linear in each factor separately:  
\[
\begin{split}
z_{1_1}\tensor\cdots \tensor (az_{i_k}+bz'_{i_k})\tensor z_{i_{k+1}}\tensor \cdots=& a(z_{1_1}\tensor \cdots \tensor z_{i_k}\tensor z_{i_{k+1}}\tensor \cdots)+ \\
&b(z_{1_1}\tensor \cdots \tensor z'_{i_k}\tensor z_{i_{k+1}}\tensor \cdots).
\end{split}
\]

Clearly every $\alpha \in \tilde{\mathfrak{A}} $ is a finite linear combination of elements in $B^{\infty}$. We define a product in $\tilde{\mathfrak{A}}$ as follows. First, for elements of $B^{\infty}$: 
\[(z_{i_1}\tensor z_{i_2}\tensor\cdots )(z'_{i_1}\tensor z'_{i_2}\tensor\cdots )=(z_{i_1}z'_{i_1}\tensor z_{i_2}z'_{i_2}\tensor\cdots )\]
We extend the product to whole of $\tilde{\mathfrak{A}}$ by linearity. Next define a norm  by 
\[\norm{\sum_{i_1,i_2,\dotsc}a_{i_1i_2\cdots}z_{i_1}\tensor z_{i_2}\tensor \cdots }=\sup\{|a_{i_1i_2\cdots}|\}\] 
It is straightforward to show that $B^{\infty}$ is an atomic basis. It follows that the above function is indeed an algebra norm and that $\tilde{\mathfrak{A}}$ is an abelian normed algebra. We also define $*$-operation by 
\[\left(\sum_{i_1,i_2,\dotsc}a_{i_1i_2\cdots}z_{i_1}\tensor z_{i_2}\tensor \cdots \right)^*=\sum_{i_1,i_2,\dotsc}\conj{a_{i_1i_2\cdots}}z_{i_1}\tensor z_{i_2}\tensor \cdots \]
It is routine to check that for $x\in \tilde{\mathfrak{A}}$, $\norm{xx^*}=\norm{x}^2$. Finally, we complete the norm and call the resulting $\cstar$ algebra $\mathfrak{A}$. The completion of a norm is a technical device that uses the fact that any normed algebra $X$ can be isometrically mapped to a norm complete algebra (a Banach algebra) $\hat{X}$ and the image $X$ is dense in $\hat{X}$ (see \cite{KR1}).\footnote{There are some delicate convergence issues here. Since $\tilde{\mathfrak{A}}$ consisting of finite sums of tensor products is dense in $\mathfrak{A}$ it often suffices to prove some statement about $\tilde{\mathfrak{A}}$ and extended it to $\mathfrak{A}$ by continuity.} With these definitions  $\mathfrak{A}$ is a $\cstar$ algebra. An important special case is when all the factor algebras $A_i=A$. We then write the infinite tensor product $\cstar$ algebra as $\inftens{A}$. Intuitively, the elements of an atomic basis $B^{\infty}$ of $\inftens{A}$ correspond to strings from an alphabet (represented by the basis $B$) with a given prefix. A general element of $A$ which is a linear combination of elements of $\inftens{B}$. Of particular interest is the 2-dimensional algebra $D$ corresponding to a binary alphabet. Thus we name $\inftens{D}$ the {\em binary} algebra. Let us fix some notation. For any finite dimensional $\cstar$ algebra $A$ the atomic basis $B^{\infty}$ for $\inftens{A}$ constructed above will be denoted by $B^{\infty}_A$ to emphasize the association. The algebras $\inftens{A}$ will be our model of signals from a source/encoder which are strings (of arbitrary length) from some alphabet. We next prove a result that is relevant for coding theory. 
\begin{propn} \label{propn:binary}
Let $A$ be an abelian $\cstar$ algebra of dimension $n$ with atomic basis $B_A=\{x_0,\dotsc,x_{n-1}\}$. Let $B_G=\{y_0,y_1\}$ be the atomic basis of the 2-dimensional algebra $G$ defined above. Then there are injective algebra homomorphisms 
\[ 
\cali{J}: \inftens{G}\rightarrow \inftens{A} \text{ and } \cali{J}': \inftens{A} \rightarrow \inftens{G}
\]
that are isometries. 
\end{propn}
\begin{proof}
We observe that it is sufficient to define an injective {\em set} map $j$ (resp.\ $j'$) from $B_G^{\infty}$ to $B_A^{\infty}$ (resp. $B_A^{\infty}$ to $B_G^{\infty}$). For we can first extend these to linear maps $\cali{J}$ (resp.\ $\cali{J}'$) on the appropriate spaces. The fact that the bases are atomic will ensure that these are injective algebra homomorphisms, in fact, isometries. Let 
\[
\begin{split}
&j(z_1\tensor \cdots \tensor z_k\tensor\unit\tensor\cdots) =\phi(z_1)\tensor \cdots \tensor \phi(z_k)\tensor\unit\tensor\cdots \text{ where }\\
& z_i\in \{y_0,y_1\} \text{ and } \phi(y_0)=x_0,\; \phi(y_1)=x_1 
\end{split}
\]
To construct $j'$ let the binary representation of the integer $n-1$ be of length $k+1$ where $k=\lfloor\log_2{n}\rfloor$. For $0\leq r \leq n-1$ $r=b^r_0+b^r_12+b^r_22^2+\cdots+b^r_k2^k$ be the binary representation of $r$ of length $k+1$ (pad it with 0's if necessary). Let $\psi:B_A\rightarrow B^{\infty}_G$ be the map defined by 
\[ \psi(x_r)= y_{b^r_0}\tensor y_{b^r_1}\tensor \cdots \tensor y_{b^r_k}\]
extend it to a map $j':B^{\infty}_A\rightarrow B^{\infty}_G$ by 
\[ j'(z_1\tensor \cdots \tensor z_k\tensor \unit\tensor\cdots )=\phi(z_1)\tensor \cdots \tensor \phi(z_k)\tensor \unit\tensor\cdots\]
The map $j'$ is injective and the proof is complete. 
\end{proof}
Let us note that from the injective maps $j\text{ and }j'$ we can construct a {\em bijective} correspondence between $B^{\infty}_A \text{ and }B^{\infty}_G$ by a Schroeder-Bernstein type construction (see \cite{Kleene}) and this can be lifted to an algebra isometry. But for us, the isomorphisms induced by maps like $j$ and $j'$ (these are certainly not unique) will be greatest interest. Essentially, what the proposition says is that it is often sufficient to restrict our attention to the special algebra $\inftens{G}$. 

The next step is to describe the state space. We recall that states of an algebra $A$ are precisely the positive functionals $\omega$ that are normalized: $\omega(\unit)=1$. Given a $\cstar$ subalgebra $V\subset A$ the set of states of $V$ will be denoted by $\mathscr{S}(V)$. Let $\mathfrak{A}=\tensor^{\infty}_{i=1} A_i$ denote the infinite tensor product of finite-dimensional algebras $A_i$. An infinite product state of $\mathfrak{A}$ is a functional of the form
\[ \Omega=\omega_1\tensor \omega_2\tensor\cdots \text{ such that }\omega_i\in \mathscr{S}(A_i)\]
This is indeed a state of $\mathfrak{A}$ for if $\alpha_k = z_1\tensor z_2 \tensor \cdots \tensor z_k\tensor\unit\tensor\unit\cdots\in \mathfrak{A}$ then 
\[\Omega(\alpha)=\omega_1(z_1)\omega_2(z_2)\cdots \omega_k(z_k), \]
a {\em finite} product. Since an arbitrary element of $\mathfrak{A}$ is the limit of sequence of finite sums of elements of the form $\alpha_k,\; k=1,2,\dotsc $ $\Omega$ is bounded by the principle of uniform boundedness. Clearly, it is positive. A general state on $\mathfrak{A}$ is a convex combination of product states like $\Omega$. 
\subsection{Analytic functions on $\cstar$ algebras}
In this section we discuss another useful construction. Let $A$ be a $\cstar$ algebra. Suppose $f(z)$ is an analytic function whose Taylor series $\sum_{n=0}^{\infty}a_n (z-c)^n$ is convergent in a region $|z-c|<R$. The convergence of the series $\sum \norm{x-c\unit}^n$ for $\norm{x-c\unit}<R$ implies that the series $\sum _{n=0}^{\infty}(x-c\unit)^n$ converges (we need completeness of $A$ for this). Thus it makes sense to talk of analytic functions on a $\cstar$ algebra. If we have an atomic basis $\{x_1,x_2,\dotsc \}$ in an abelian $\cstar$ algebra then the functions are particularly simple in this basis. Thus if $x=\sum_i a_ix_i$ then $f(x)=\sum_i f(a_i)x_i$ provided that $f(a_i)$ are defined in an appropriate domain. We will mostly take this as our definition with the understanding that the constant function $c$ is identified with $c\unit$. 
\section{Algebraic approach to probability}\label{sec:prob}
We have observed that discrete signals from a source are modeled by an abelian algebra. The elements of the algebra correspond to random variables representing the output of the source. With random variables we always associate a probability distribution. In the standard treatment of probability theory the probability or sample space is introduced first. Random variables are defined as (measurable) real (or complex, in general) functions on this space. One then finds the probability distributions of the random variables and most important quantities like mean, variance and correlations are based on these distributions. In particular, the mean or expectation value plays a central role. Note that random variables can be added and multiplied making it a real algebra (scalars are the constant random variables). Note also that random variables also represent quantities that are actually measured or observed- the voltage across a resistor, the currents in an antenna, the position of a Brownian particle and so on. The probability distribution corresponds to the {\em state} of the devices that produce these outputs. We will take the alternative view and start with these observables as our basic objects. In this way, we single out the objects which are relevant to a specific problem. In the following paragraphs we formalize these notions. 

\subsection{Basic notions}\label{sec:prob1}
\newcommand{\scrp}[1]{{\mathscr #1}}
\newcommand{\intd}{\mathrm{d}}
A classical observable algebra is an abelian complex $\cstar$ algebra $A$. It is convenient to use complex algebras. We can restrict our attention to real algebras whenever necessary. Recall that a state on $A$ is positive linear functional $\omega$ such that $\omega(\unit)=1$. We can identify $\omega$ with a {\em probability measure} as follows. Suppose $(M,\cali{S},P)$ is probability space, ($M$= sample space, $\cali{S}$ = $\sigma$-algebra, $P$=probability measure). Let $L_{\infty}(M,\cali{S},P)$ (or simply $L_{\infty}(M)$ if the measure structure is clear) be the set of essentially bounded measurable complex functions.\footnote{A function $f$ is said to be essentially bounded if there is a constant $K$ such that $|f(x)|\leq K$ almost everywhere. The essential is the infimum over all such $K$: $\mathrm{ess}\;\sup(|f|)= \inf \{k: P\{x:|f(x)|>k\}=0\}$.} We can give it a $\cstar$ structure as in the case of $C(X)$, the space of continuous functions on a compact topological space $X$ (see equation (\ref{eq:exampBasic1})), but using the essential supremum instead of the ordinary supremum. If $B\in \cali{S}$ then the {\em indicator} function $I_B\in L^{\infty}(M,\comp)$ and 
\[\int_M I_B \intd P = P(B)\]
where the integral is defined in the sense of Lebesgue. Note that $\omega_P(f)\equiv  \int fdP$ is a positive linear functional on $L_{\infty}(M)$. Since $\omega_P(\unit)=1$ it is a state. 

\noindent
\begin{Def}
A {\em probability algebra} is a pair $(A, S)$ where $A$ is an observable algebra and $S\subset \scrp{S}(A)$ is a set of states. A probability algebra is defined to be {\em fixed} if $S$ contains only one state. 
A probability algebra $\scrp{A}_1=(A_1,S_1)$ is defined to be a {\em cover} of another $\scrp{A}_2=(A_2,S_2)$ if there is an algebra homomorphism $\phi:A_1\rightarrow A_2$ and a one-to-one correspondence $\gamma:S_1\leftrightarrow S_2$ such that the following conditions hold: i.\ $\phi$ is onto and ii.\ for all $x\in A_1$ and $\omega\in S_1$: $\omega(x)=\gamma(\omega)(\phi(x))$. 
\end{Def}
Let $\omega$ be a state on an abelian $\cstar$ algebra $A$. Call two elements $x,y\in A$ uncorrelated in the state $\omega$ if $\omega(xy)=\omega(x)\omega(y)$. Note that this definition depends crucially on the state: the same two elements can be correlated in some other state $\omega'$. Two natural questions are immediate. Are there any states for which {\em every} pair of elements of $A$ are uncorrelated? Are there a pair of elements which are uncorrelated in {\em every} state? Two trivial candidates for the second question are $\unit$ and 0. Either of them is uncorrelated to every element. We implicitly exclude these two trivial cases. Concerning the second question the answer is negative in general. On the first question, a state $\omega$ is called multiplicative if $\omega(xy)=\omega(x)\omega(y)$ for all $x,y\in A$. Note that the notion of positivity defines a partial order on the space of functionals making it an ordered vector space \cite{KR1}. The set of states, $\mathscr{S}$, is convex in the usual sense that for numbers $p_i\geq 0,\; \sum_{i=1}^kp_i=1$ and states $\omega_i,\;i=1,\dotsc,k$ the functional $\sum_i p_i\omega_i$ is also a state. The extreme points of $\mathscr{S}$ are called {\em pure} states. In the case of abelian $\cstar$ algebras a state is pure if and only of it is multiplicative \cite{KR1}. Thus in a pure state any two observables are uncorrelated. This is not generally true in the non-abelian quantum case. 

Next we come to the important notion of {\em independence}. First, given $S\subset A$ let $A(S)$ denote the subalgebra generated by $S$ (the smallest subalgebra of $A$ containing $S$). Two subsets $S_1,S_2\subset A$ are defined to be {\em independent} if all the pairs $\{(x_1,x_2): x_1\in A(S_1), x_2\in A(S_2)\}$ are uncorrelated. As independence and correlation depend on the state we sometimes write $\omega$-independent/uncorrelated when to emphasize this. Clearly, independence is much stronger condition than being uncorrelated. It is easy to construct examples in 3 or more dimensions where a pair of observables $x,y$ are uncorrelated but they are not independent: for example, $x^2 \text{ and } y$ maybe correlated. However, in 2 dimensions $x\text{ and }$ are uncorrelated if and only if one of them is 0 or $c\unit$. Let us note that as in the quantum case two dimensions is an exceptional case. The next theorem shows the structural implications of independence. 
\begin{thm} \label{thm:structIndep}
Two sets of observables $S_1,S_2$ in a finite dimensional abelian $\cstar$ algebra $A$ are independent in a state $\omega$ if and only if for the (unital) subalgebras $A(S_1)$ and $A(S_2)$ generated by $S_1$ and $S_2$ respectively there exist states $\omega_1 \in \scrp{S}(A(S_1)),\; \omega_2\in \scrp{S}(A(S_2))$ such that $(A(S_1)\tensor A(S_2),\{\omega_1\tensor\omega_2\})$ is a cover of $(A(S_1S_2),\omega')$ where $A(S_1S_2)$ is the subalgebra generated by $\{S_1,S_2\}$ and $\omega'$ is the restriction of $\omega$ to $A(S_1S_2)$. 
\end{thm}
\begin{proof}
First assume that $S_1=\{x\}$ and $S_2=\{y\}$. 
Let $\{x_1,\dotsc,x_n\}$ be an atomic basis of $A$. Let $x=\sum_ia_ix_i$ and $y=\sum_ib_ix_i$. Some of these coefficients may be 0 and some may be equal. Write 
\[ x=a_1P_1+a_2P_2+\cdots +a_kP_k \text{ and } y=b_1Q_1+b_2Q_2+\cdots+b_lQ_l\] 
Here the $a_i$'s are distinct the $P_i=x_{i_1}+x_{i_2}+\cdots+x_{i_r}$ corresponding to all basis elements whose coefficients are equal to $a_i$. Similarly for $Q_j$'s. 
Note that $P_iP_m=\delta_{im}$ and $Q_jQ_s=\delta_{js}$. By Lagrange interpolation there are  polynomials $f_i(\lambda),\;i=1,\dotsc,k$ and $g_j,\;j=1,\dotsc,l$ such that $f_i(a_r)=\delta_{ir}$ and $g_j(b_s)=\delta_{js}$. Since $x,y$ are $\omega$-independent 
\beq\label{eq:indStruct}
\omega(f_i(x)g_j(y))= \omega(P_iQ_j)=\omega(P_i)\omega(Q_j). 
\eeq
The subalgebra $A(S_1)$($A(S_2)$) is generated by the $P_i$'s($Q_j$'s). Clearly $\{P_i:i=1,\dotsc,k\}$ and $\{Q_j:j=1,\dotsc,l\}$ are atomic bases for $A(S_1)$ and $A(S_2)$ respectively. Define states $\omega_1$ and $\omega_2$ of $A(S_1)$ and $A(S_2)$ resp.\ by restricting $\omega$ to these subalgebras. Let $\phi:X\tensor Y\rightarrow A'$ be the natural map $\phi(u\tensor v)=uv$. Using equation \ref{eq:indStruct} it is a routine check that $(A(S_1)\tensor A(S_2),\{\omega_1\tensor\omega_2\})$ is a cover of $(A(S_1,S_2),\omega')$. 

Now for the general case. Since $A(S_1)$ and $A(S_2)$ are subalgebras of $A$ they have atomic bases $\{u_i\}$ and $\{v_j\}$ respectively. As in the previous case we have polynomials $\{p_i\}\text{ and }\{q_j\}$  in several variables such that $p_i(x_1,\dotsc, x_{k_i})=u_i$ and $q_j(y_1,\dotsc,y_{m_j})=v_j$ where \(x_i\in S_1\text{ and }q_i\in S_2\). We do not have easy interpolating polynomial in this case. By repeating the argument of the singleton case above we get the appropriate cover and complete the proof. 

The converse is clear from the definition of a cover and the fact that in a {\em product state} $\omega_1\tensor\omega_2(z_1\tensor z_2)=\omega_1(z_1)\omega_2(z_2)$. 
\end{proof}
We can even extend it to infinite tensor product by restricting to finite segments. The next step is to extend the notion of independence to more than two subsets. Let $S_1,\dotsc,S_k\subset A$ and $\omega$ a state of $A$. Then the subsets are defined to be $\omega$-independent if for all $x_i\in A(S_i),\; i=1,\dotsc, k$ we have 
\[ \omega(x_1\cdots x_k)=\omega(x_1)\cdots \omega(x_k)\]
Here $A(S_i)$ is the subalgebra generated by $S_i$. We can then show that for states $\omega_i\in \scrp(A(S_i))$, the restriction of $\omega$ to $A(S_i)$ the pair 
\((A(S_1)\tensor\cdots \tensor A(S_k), \omega_1\tensor\cdots\tensor \omega_k)\) is a cover of $A(S_1\dotsc S_k),\omega'$, where $\omega'$ is the restriction of $\omega$ to $A(S_1\dotsc S_k)$, the algebra generated by $S_1,\dotsc, S_k$. 
We thus see the relation between independence and (tensor) product states in the classical or commutative theory. The non-commutative or quantum case is more delicate and requires careful handling. 
\subsection{Probability distribution functions}
In this section we investigate another important concept of a (cumulative) distribution function ({\bf d.f}) in the algebraic framework. As the paper's primary concern is an alternative formulation of mathematical models of information and communication we do not undertake an extensive exploration of the algebraic approach to probability concepts. However, the notion of a distribution function underpins large part of probability theory and its applications. One of the advantages of using $\cstar$ or more general Banach algebra is that we have both algebraic and analytical methods at our disposal.  

Given a subalgebra $B\subset A$ of an abelian  $\cstar$ algebra let $S_a= \{x\in A: xs=0\;\forall s\in S\}$ be the {\em annihilator} of $S$.  This is an ideal\footnote{An ideal of a algebra $A$ is a subset $I$ of $A$ which is closed under addition and for every $x\in A$, $xI\subset I$. Hence a non-zero proper ideal cannot contain the identity of $A$} and hence there is an {\em approximate identity}. An approximate identity in an ideal $B$ is a {\em net} $\{y_\lambda\}$ with $0 \leq y_{\lambda}\leq \unit$ such that $xy_{\lambda} \rightarrow x $ (also $yx_{\lambda}  \rightarrow x,\;\forall x\in B$  if the algebra is nonabelian). For the details see \cite{KR1}. Obviously $S_a$ cannot contain the identity of the original algebra unless $S=\{0\}$. We only mention that nets \cite{Kelley} are generalization of sequences where the indexing set is not required to be countable. However, in the case of separable algebras (algebras with a dense countable set) the reader may substitute ``sequence'' for ``net''. In the following it will suffice for our purpose to restrict to the separable case although we often use the language of ``nets''. We can now define distribution of a set of observables. 

\begin{Def}
Let $S=\{x_1,x_2,\dotsc, x_n\}$ be a finite self-adjoint subset of $A$ where $(A,\omega)$ is a fixed probability algebra. For ${\tt t}=(t_1,t_2,\dotsc, t_n )\in \real$ let $S_{\tt t}\subset A$ denote the set of elements  $\{(t_i\unit - x_i):i=1,\dots, n\}$ and $S_{\tt t}^-$ the set of elements  $\{z_-:z\in S_{\tt t}\}$, negative parts of members of $S_{\tt t}$. Let $\{e_{\lambda}({\tt t})\}$ be approximations of identity in the annihilator ideal $(S^-_{\tt t})_a$. Then the $\omega$-distribution of $S$ is defined to be the real function 
\[ F_S({\tt t})= \lim_{\lambda} \omega(e_{\lambda}) \]
\end{Def}
The rationale for this definition is simple. For convenience, restrict to a single random variable. Suppose $X$ is a bounded random variable on a probability space $\{\Omega,\cali{S},P\}$. Then the distribution function $f(t)=P(\{\alpha\in \Omega: X_t=tI-X(\alpha)\geq 0\})$. For a fixed $t$ write the random variable $X_t=X_{t+}-X_{t-}$ as the difference of two non-negative random variables. Then the distribution function of $X$ is the probability of the event $E_t$ where $E_t=\{ \alpha \in \Omega: X_t(\alpha) \geq 0\}$. Consider now $X_{t-}$ and $G_t=\{\alpha: X_t(\alpha)<0\}=\Omega-E_t$. Then $X_{t-}$ is $>0$ on $G_t$ and 0 outside it. If $Y$ is any function on $\Omega$ such that $YX_{t-}=0$ then $Y$ must vanish on $G_t$. Conversely any function $Y$ that vanishes on $G_t$ satisfies the equation $YX_{t-}=0$. In particular the indicator function $\scrp{I}_{F_t}$ satisfies it. The function $\scrp{I}_{F_t}$ is the identity on $(X_{t-})_a$ and its expectation value $\int \scrp{I}_{F_t}\intd P=P(F_t)$. 
Although, the indicator functions are not generally continuous we can approximate them by a sequence of continuous functions. This sequence is an approximate identity in the $\cstar$ algebra of continuous functions. In most cases of interest to us the algebras will be separable. Then the nets can be replaced by sequences. Note that since the net $\{e_\lambda\}$ is bounded and increasing the net $\{\omega(e_\lambda)\}$ converge. Finally, let us observe that even though the approximate identity is not unique the distribution function as defined above is unique. To prove this  $\{e_\lambda\}, \{f_\lambda\}$ are two approximate identities. Then using the fact $\omega(e_\lambda f_\mu -e_{\lambda'} f{\mu '})=\omega(f_\mu(e_\lambda -e_{\lambda'})+e_{\lambda'}(f_\mu -f_{\mu'}))$ is Cauchy since $f_\mu(e_\lambda -e_{\lambda'})\rightarrow (e_\lambda -e_{\lambda'})$ and $e_{\lambda'}(f_\mu -f_{\mu'})\rightarrow f_\mu -f_{\mu'}$ we conclude that the double-net $\{\omega(e_{\lambda}f_\mu)\}$ converges to the limit $\lim_{\lambda}\omega(e_{\lambda})=\lim_{\mu}\omega(f_{\mu})$. Extending the definition of the d.f to an arbitrary element $z$ in the algebra is simple. Write $z=x+iy$ where $x$ and $y$ are self-adjoint. Let $F_x(t)\text{ and }F_y(t)$ denote the d.f of $x$ and $y$ respectively. Then the d.f of $z$: $F_z(t)=F_x(t)+iF_y(t)$. 

%The definition of probability distribution is applicable in any abelian $\cstar$ algebra. In the finite-dimensional case 
\begin{thm}\label{thm:dist}
Let $x_1,\dotsc, x_n$ be self-adjoint elements of an abelian $\cstar$ algebra $A$. Let $F(t_1,\dotsc, t_n)$ be their joint distribution function. Then $F(t_1,\dotsc, t_n)$  is non-negative, left-continuous  and non-decreasing in each variable. We also have boundary conditions 
\[ \lim_{t_1, \dotsc, t_n \rightarrow \infty} F(t_1,\dotsc, t_n)=1 \text{ and } \lim_{t_1, \dotsc, t_n \rightarrow -\infty} F(t_1,\dotsc, t_n)=0\]
If the elements are independent and $F(t_i)$ denotes the distribution function of $x_i$ then
\[F(t_1,\dotsc, t_n)=F(t_1)F(t_2)\cdots F(t_n).\]
If a sequence $x_n\rightarrow x$ in then the corresponding d.f's $F_{x_n}(t)\rightarrow F_x(t)$. 
\end{thm}
%Note: rework the proof using spectral theorem.
\begin{proof}
This is of course a standard result in probability theory. We sketch an algebraic proof in the current setting. The most direct approach is to use the notion of continuous function calculus which essentially asserts that continuous functions on the spectrum  can be lifted to define functions on the algebra. More precisely, given an element $x\in A$ there is an isometric algebra homomorphism between the algebra of continuous functions on the spectrum of $x$, $C(\text{sp}(x))$ and the closed subalgebra  $C(x)$  generated by $x$ \cite{KR1}. Thus for every function $f(u)$ on $\text{sp}(x)$ there is a unique element $f(x)$ in $C(x)$ such that 
if $f(u)\geq 0$ then $f(x)\geq 0$. Since for any real $c$ and $\delta>0$ , \(|t+\delta-u|-(t+\delta-u)\leq |t-u|-(t-u)\) we infer that \(|t+\delta-x|-(t+\delta-x)\leq |t-x|-(t-x)\) for self-adjoint $x\in A$. Now for any $y\in A$ if $xy=0$ then $|x|y=0$ and hence $x_+y=x_-y=0$. So if $x\leq z$ and $ v\in A$ then $zv=0$ implies $xv=0$. Thus the annihilator ideal of $|t+\delta-x|-(t+\delta-x)$ contains the annihilator ideal of $|t-x|-(t-x)$. The continuity follows from the following construction which is useful for calculating distributions. Write $x(t)= t\unit -x$, ${\tt t}=(t_1,t_2,\dotsc, t_n )$ and \( \chi({\tt t})= x_1(t_1)_+ \times x_2(t_2)_+ \times \dotsb \times x_n(t_n)_+\). For integer $m>0$ let 
\beq \label{eq:dist}
e_m({\tt t}+1/m) = m\chi({\tt t}+1/m)(1+m\chi({\tt t}+1/m))^{-1}\equiv\frac{m\chi({\tt t}+1/m)}{1+m\chi({\tt t}+1/m)}
\eeq
where ${\tt t}+1/m= (t_1+1/m,t_2+1/m,\dotsc, t_n+1/m )$. Although $e_m({\tt t}+1/m)$ is {\em not} a member of the annihilating ideal $S^-({\tt t})_a$ of $S^-({\tt t})$ it belongs to $S^-({\tt t}+1/m)_a\supset S^-({\tt t})_a$. Let $e_{\lambda}({\tt t})$ be an approximate identity in $S^-({\tt t})_a$. One can show using the Gelfand representation that 
\[\lim_{\lambda}\omega(e_\lambda({\tt t}))=\lim_{m\rightarrow\infty} \omega(e_m({\tt t}))\]
We omit the details but the reader can convince herself by taking an algebra of functions. 
%Should we try double sequence? 
 
This implies the first part of the theorem. To prove the boundary conditions we use the fact that the spectrum of any element $x\in A$ is bounded by $\norm{x}$. Hence, for $t<-\norm{x}$, $t\unit-x$ has a strictly negative spectrum. Then $(t\unit-x)_-=-(t\unit-x)$ is invertible and its annihilator ideal consists of $0$ alone. Consequently, $F(t,\dotsc, )=0$ for all $t<-\norm{x}$. The other extreme case is proved similarly, $t\unit -x$ being strictly positive for $t>\norm{x}$. Finally, suppose the elements $\{x_1,x_2\dotsc,x_n\}$ are independent. Since $x_{+}$ lies in the closed subalgebra generated by $x$ the definition of independence and equation \ref{eq:dist} implies that the joint distribution function is a product. One proves the last statement using a sequence like (\ref{eq:dist}). 
\end{proof}
We see that, starting from a purely algebraic definition of independence and distributions we can recover their essential properties. In particular, for algebras which are finite or infinite tensor product of finite-dimensional algebras we have the following. 
\begin{propn} \label{prop:struct_annIdeal}
Let $A$ be a finite-dimensional abelian $\cstar$ algebra. Let $x\in \tensor^{\infty} A$ and $x_a$ its annihilating ideal. Suppose $x$ is a finite sum. Then there is a unique (up to permutation) decomposition 
\[ x=\sum a_i P_i \text{ such that } P_iP_j=\delta_{ij}P_j \text{ and } a_i \neq 0 \text{ distinct }\]
Further, there exist polynomials without constant term $g_i$ such that $P_i=g_i(x)$. Thus, $x=\sum_i a_ig_i(x)$. 
Then $x_a$ has an identity $\unit -\sum_i P_i$. 
\end{propn}
\begin{proof}
Since $x$ is finite sum it may be considered as an element of $\tensor^n A$ for some finite $n$. The space $\tensor^n A$ has a finite atomic basis, say, $\{Y_1, \dotsc, Y_m\}\; (m=2^{\text{dim}(A)})$. Let $x=\sum_{i=1}^m a_i Y_i$ and let $J=\{i: a_i=0\}$. Then $x=\sum_{i\notin J} a_iY_i$. Let $P_i$ be the sum of all $Y_i$ for which the coefficients $a_i$ are equal. then $x=\sum_i a_iP_i$ with $a_i$ distinct and non-zero. Next use Lagrange interpolation to obtain polynomials $g_i$ such that $g_i(0)=0$ and $g_i(a_j)=\delta_{ij}$. To prove uniqueness let $x=\sum_j b_j Q_j$ be another such decomposition. Then $xP_iQ_j=a_iP_iQ_j=b_jP_iQ_j$. Since $\sum_i P_i x=x$ for a fixed $i$ there must be at least one  $j_i$ with $P_iQ_{j_i}\neq 0$ then $a_i=b_{j_i}$. There cannot be more than one such $j_i$ since the $b_j$'s are distinct. Arguing in the reverse direction we conclude that $i\leftrightarrow j_i$ is a permutation. The last statement follows trivially. 
\end{proof}
Let $x=\sum_i a_i P_i$ be as in the proposition. We call this the spectral decomposition of $x$. If $\omega$ is a state define 
\[ \scrp{I}_{\omega}(x)= \sum_i \omega(P_i) P_i\]
The map $\scrp{I}_{\omega}(x)$ can be considered as a ``centroid'' of the possible outcomes of measurement of $x$. 
We can extend the proposition to arbitrary element in $\scrp{A}=\tensor^{\infty}A$ by using a sequence of finite-dimensional projections as above to approximate. However, the proposition suffices for most of our requirements. Now let 
\[Z=\sum_{k=1}^{\infty} X_k,\; X_k \in \tensor^k A\]
$Z$ may {\em not} be a member of $\scrp{A}$ in general as we treat the above as a formal sum. However, we suppose that for real $t$, $(t\unit -Z)_+=(|t\unit-Z|+ (t\unit -Z))/2$ can be expressed as finite sum.  We will see an example below. Then the required identity is given as follows. It is clear that for $\delta>0$ small enough $|(t+\delta)\unit-Z|+ ((t+\delta)\unit -Z)=\sum_ka_k Y_k :a_k >0$ is finite sum where $Y_k$ constitute an atomic basis. Let $P_\delta= \sum_k Y_k$. Then the required identity is given by $P_0= \lim_{\delta \rightarrow 0} P_\delta$. This is essentially a variant of equation \ref{eq:dist} in Theorem \ref{thm:dist}. 
\subsection{Examples}\label{sec:examples}
In this section we consider some examples from standard probability theory. It will be demonstrated that the algebraic approach not only gives a different perspective on some familiar situations it can also provide additional computational tools. First, we review the correspondence between some concepts from the standard theory with our algebraic model. An event in probability theory is a measurable subset of the probability space. The random variable characterizing any (measurable) subset $S$ is its indicator function $I_S$. In the algebraic language it is a projection $Q_S$. The probability of the event corresponds to the expectation value $\omega(Q_S)$ of the projection. In the cases we consider the projections will generally exist in the algebra itself. In some cases we consider infinite formal sums which are {\em not} in the algebra but any finite segment of the sum do belong to the algebra. In the actual computation we always use a ``cut-off'' to restrict to such a finite segment. In the cases where projections are not members of the algebra we can find a sequence (or net) that ``converges in the mean'' to the appropriate projection or indicator function. This situation generally arises in the continuous case which is only touched upon peripherally.   
\be
\item
{\bf Binomial distribution.} Consider again infinite sequences of Bernoulli trials as in the second example of the previous section. We can think of coin-tossing with ``heads'' signaling success. Let $Z$ be the observable (random variable) corresponding to the number of success. What is its d.f.?  Let $n,k$ be a positive integers with $k<n$. We want to find the distribution $F(k:n)$ of $Z$. Recall that $G$ is the 2-dimensional algebra and let $A=\tensor^{n}G$. Let $\{y_0,y_1\}$ be the atomic basis of $G$ with $y_1$ corresponding to success. Set 
\[ 
\begin{split}
Z= &\sum_{\cali{S}} y_1\tensor y_0\tensor\cdots\tensor y_0+\sum_{\cali{S}} 2 y_1\tensor y_1\tensor y_0\cdots\tensor y_0+\cdots + \\
   & \sum_{\cali{S}} r\underbrace{y_1\tensor y_1\tensor \cdots \tensor y_1}_r\tensor \underbrace{y_0\tensor y_0\tensor \cdots \tensor y_0}_{n-r}+\cdots + n
   y_1\tensor y_1\tensor \cdots \tensor y_1 \\
   =& \sum_{r=1}^n r Y_r
\end{split}
\]
Here $\cali{S}$ denotes the distinct permutations of the factors in the tensor product. Thus, the $r$th term $Y_r$ is the sum of all $\binom{n}{r}$ products with $r$ $y_1$'s. Its value is $r$. Note that $Y_rY_s=\delta_{rs}$. We have 
\[U=|k\unit - Z|-(k\unit -Z)=  \sum_{r=k+1}^n rY_r\]
In this case the identity in the annihilator ideal of $U$ exists and is given by the projection operator $P=\sum_{r=0}^k Y_r$. Since the Bernoulli 
$ F(k:n)=\Omega(P)=\sum_{0}^k \binom{n}{k} p^k(1-p)^{n-k}$. Note that we can easily find the distribution in states where the observables are not independent. 
\item
{\bf Waiting time.}
Let us start with a simple version of the problem of waiting time. Suppose we have a binary source with fixed probability distribution emitting a bit per unit time. The waiting time is the time elapsed before the first appearance of 1. It is a random variable or observable $W$ in our formalism. Using the notation above 
\[ W= y_0\tensor y_1\tensor \unit\tensor\cdots+ 2y_0\tensor y_0\tensor y_1\tensor \unit\tensor\cdots+3 y_0\tensor y_0\tensor y_0\tensor y_1\tensor \unit\tensor\cdots+\cdots\]
This is an unbounded infinite sum and does not belong to the algebra. However, for any $t\geq 0$, 
\[
\begin{split}
&F_W(t)\equiv \frac{|t\unit - W|+t\unit-W}{2}  =  \\
& ty_1\tensor \unit+(t-1)y_0\tensor y_1\tensor\unit+\cdots+(t-\floor{t})\underbrace{y_0\tensor \cdots\tensor y_0}_{\floor{t}\text{ factors }}
\tensor y_1\tensor\unit 
\end{split}
\]
is finite (of course, $F_W(t)=0\text{ for } t<0$). 
Here $\floor{t}$ is the largest integer $\leq t$. Using the trick explained before the examples we replace $t$ by $t+\delta$ (this is to take into account the case when $t$ is an integer). The required projection (approximate identity) is 
\[ P_W(t)=  y_1\tensor \unit+y_0\tensor y_1\tensor\unit+\cdots+\underbrace{y_0\tensor \cdots\tensor y_0}_{\floor{t}\text{ factors }}
\tensor y_1\tensor\unit \]
The distribution function in a state $\Omega$ is given by $F_W(t)=\Omega(P_W(t))$. If $\Omega=\omega\tensor\omega\tensor\cdots$ is an infinite product state with $\omega(y_1)=p=1-\omega(y_0)$ then $F_W(t)=\sum_{k=0}^{\floor{t}} p(1-p)^{k}$.  

Next we generalize the problem of waiting time to arbitrary strings. Explicitly, given a string $\xi$ of length $n$ the waiting time is the time before a contiguous stream of bits matching $\xi$ appears. The preceding case is for $\xi=1$. We will only construct the observable corresponding to waiting time $W$ in this general case. It gives a nice illustration of the algebraic techniques. Let $X$ be the tensor representation of $\xi$. Waiting time 0 corresponds to the observable $X\tensor \unit$. We use the following notation. Write $\unit_1$ for the identity in the 2-dimensional space $G$ and $\unit_k=\unit_1\tensor\unit_1\tensor\cdots\unit_1$, the $k$-fold tensor product. The symbol $\unit$ (without subscripts) will be reserved for the identity in $\tensor^{\infty} A$. The element $Y_0=X\tensor \unit$ corresponds to waiting time 0: the first $n$ symbols received match the given string. We expect the element corresponding to waiting time 1 will be ``proportional'' to $Y_1'=\unit_1\tensor X\tensor \unit$. Although  $Y_0$ and $Y_1'$ are projections they need not be orthogonal in the sense $Y_1'Y_0=0$. So they do not correspond to mutually exclusive events. Recall that when interpreted as functions on some measure space projections are indicators of measurable sets (events). We therefore adopt an orthogonalization scheme similar to Gram-Schmidt. The observable $Y_1=Y_1'-Y_1'Y_0$ is projection and satisfies $Y_1Y_0=0$. Viewed as a function it takes value 1 only when the input string is of the $\zeta=b_0\xi\ldots$ and such that the prefix of length $n$ of $\zeta$ does not match $\xi$. It corresponds to waiting time 1. Defining inductively, let 
\[ 
\begin{split}
Y_m &= Y_m'-Y_m'(Y_0+Y_1\cdots+Y_{m-1}) \\
&= \unit_{m}\tensor X\tensor\unit - \unit_{m}\tensor X\tensor\unit(Y_0+Y_1\cdots+Y_{m-1})
\end{split}
\]
It is easily verified that $Y_jY_k=\delta_{jk}Y_k$. The element $W=\sum_{k=0}^\infty kY_k$ corresponds to the waiting time in this case. Again it is not an element of the algebra but $|t\unit- W|+t\unit-W$ is. 
\item
{\bf Markov Chains.} We define a discreet time Markov chain on an observable algebra $(A,\omega)$ as a sequence of {\em positive} and {\em unital} maps $\{\phi_0,\phi_1, \dotsc, \}$ and an initial element $x_0\in A$. Let us confine to discrete chains. Let $\cali{A}=\{x_1, x_2, \dotsc, \}$ be a fixed atomic basis. A chain-state is a sequence $\{z_0,z_1,\dotsc,\}$ where each $z_i\in \cali{A}$. The usual term for what we call chain-state is simply ``state'' but the latter has a very specific meaning in operator algebras. Let $\xi_n=\{z_0, z_1, \cdots, z_n\}$ be a finite segment of the chain-state. We are interested in the transition from $x_0$ to $x_n$ via the path $\xi_n$. The transition probability is defined recursively as follows. 
\[
\begin{split}
&y_1=\phi_0(z_0),\quad y_k= \phi_{k-1}(z_{k-1}y_{k-1})\text{ and }\\
&\text{transition probability }p(z_0\xrightarrow{\xi_n} z_n)=\omega(z_ny_n) \\
\end{split}
\]
Let us examine this definition in the special case of stationary Markov chains. A Markov chain is defined to be stationary if all the transition maps are identical: $\phi_0=\phi_1=\phi_2=\dotsb $. For a stationary chain
\[
\begin{split}
p(z_0\xrightarrow{\xi_n} z_n)& = \omega(z_n\phi(z_{n-1}\phi(z_{n-2}\phi(\dotsm z_1\phi(z_0)))))\\
& =\omega(z_0)\phi(i_n,i_{n-1})\phi(i_{n-1},i_{n-2})\dotsm \phi(i_1,i_0) \\
\end{split}
\]
\ee
Here $\phi(i,j)$ is the $(ij)$th matrix element of $\phi$ with respect to the basis $\cali{A}$ and $z_k=x_{i_k}$. This looks very similar to quantum transition probability. In the later case the$x_i$ are projections on a Hilbert space. Further, when we consider transitions over all possible paths then we get an analogue of Feynman's ``sum over paths'' for total transition probability. 
\subsection{Limit theorems}
The limit theorems of probability theory are important for its theoretical structure as well as its empirical justification. We will be primarily concerned with the bounded case where the proofs are simpler. We state two of these but prove only the {\bf weak law of large numbers}. From information theory perspective it is perhaps the most useful limit theorem. Let $X_1,X_2,\cdots, X_n$ be independent, identically distributed (i.i.d) random variables on a probability space $\Omega$ with probability measure $P$. Let $\mu$ be the mean of $X_1$ (hence any $X_i$). We assume that the the variance $E(X_1-\mu)^2$ is bounded. Here, $E(X)$ denotes the expectation value of random variable $X$. 

\bi
\item {\bf Weak law of large numbers}. Given $\epsilon>0$ 
\[\lim_{n\rightarrow \infty}P(|S_n=\frac{X_1+\cdots +X_N}{n}-\mu |>\epsilon)=0\]
\item {\bf Central limit theorem}. If $0<E(X_1^2)=\sigma<\infty$ then for any real $x$ as $n\rightarrow \infty$
\[P(\frac{S_n}{\sqrt{n}} \leq x)\rightarrow \Phi(x)= \frac{1}{\sqrt{2\pi}}\int_{-\infty}^x \exp{(-(t-\mu)^2/2\sigma)} \intd x\]
\ei
A few comments about these famous limit theorems. These are statements about different types of convergence \cite{Billingsley}. The theorems can be strengthened but since we are dealing with bounded random variables the above formulations suffice. These theorems require assigning of probabilities. All we have at our disposal is the algebra and one or more positive functionals (states) which give us expectation values. But we have already seen how to define probability distribution functions. What we need are appropriate projections or approximations to them. Given a self-adjoint observable $x$ and a real number $a$ write $x-a$ for the element $x-a\unit$. Let $A(x-a)_+$ be the (two-sided) ideal generated by the positive part of $x-a$. Let $e_n=(x-a)_+[(x-a)_+ +\delta_n]^{-1}$ where $0<\delta_n$ such that $\lim_{n\rightarrow \infty}\delta_n=0$. Then it can be shown that for any $y\in A(x-a)_+$, $\lim_{n\rightarrow \infty} e_ny\rightarrow y$ in the norm. Hence, $\{e_n\}$ is an increasing {\em sequence} approximating identity (see \cite{Tak1}). We write $\mathbb{ P}(x>a)$ for this approximate identity in $A(x-a)_+$. It is not unique but that does not matter since all the limits that we use it to define are independent of the particular choice. The probability corresponding to the ``event'' $x>a$ is defined to be $P(x>a)=\omega(\mathbb{ P}(x>a))=\lim_{n\rightarrow \infty} \omega(e_n)$. Similarly we can define $P(x<a)=\omega(\mathbb{ P}(x<a))$ where $\mathbb{ P}(x<a)=\{f_n\}$ is an approximate identity in the ideal $A(x-a)_-$ obtained by replacing $(x-a)_+$ by $(x-a)_-$ in $e_n$. We can define more complicated events by algebraic operations but it is not necessary for what follows. We also note that although we use probabilistic language in the statements of the results below all the expressions are actually defined in a strictly algebraic setting without reference to any underlying probability space. 
\begin{lem} [Chebysev inequality]
Let  $x, y \in A$ be self-adjoint where $(A,\omega)$ is an observable algebra and $y\geq 0$. For any number $\epsilon>0$ we have 
\[P(y>\epsilon) \leq \frac{\omega(y)}{\epsilon} \text{ and } \]
\[P(|x-\omega(x)| >\epsilon ) \leq \frac{\omega([x-\omega(x)]^2)}{\epsilon^2} \]
\end{lem}
\begin{proof}
Let $\{e_n\}$ be an approximate identity in the ideal $A(y-\epsilon)_+$. By definition $e_n\leq \unit$. Hence, \(\omega(y)=\omega(ye_n)+\omega(y(\unit-e_n))\geq \omega(ye_n)\). Since $(y-\epsilon)_-$ annihilates the ideal $A(y-\epsilon)_+$, \(\omega(ye_n)=\omega([y-\epsilon]e_n)+\omega(\epsilon e_n)=\omega([y-\epsilon]_+e_n)+\epsilon\omega(e_n)\geq \epsilon\omega(e_n)\). Hence, $\omega(y)\geq \epsilon\omega(e_n)$. Taking limits we obtain the first inequality. 
Observe that for any $x\in A$, $P(|x|>\epsilon)=P(|x|^2>\epsilon^2)$ for the ideals $A(|x|-\epsilon)_+$ and  $A(|x|^2-\epsilon^2)_+$ coincide. This follows from the identities $|x|^2-\epsilon^2= (|x|+\epsilon)(|x|-\epsilon)$ and hence $(|x|^2-\epsilon^2)_+= (|x|+\epsilon)(|x|-\epsilon)_+$ plus the fact that $|x|+\epsilon$ is invertible. Hence the second inequality follows from the first by putting $y=(x-\omega(x))^2$ and using $\epsilon^2$ in place of $\epsilon$. 
\end{proof}
We will prove next a convergence result which implies the weak law of large numbers. 
\begin{thm} [Law of large numbers (weak)] \label{thm:weak-law} 
If $x_1,\dotsc,x_n,\dotsc$ are  
$\omega$-\\independent self-adjoint elements in an observable algebra and $\omega(x_i^k)=\omega(x_j^k)$ for all positive integers $i,j\text{ and }k$ (they are identically distributed) then 
\[\lim_{n\rightarrow \infty} \omega(|\frac{x_1+\dotsb+x_n}{n}-\mu|^k)=0 \text{ where } \mu=\omega(x_1) \text{ and } k>0\]
\end{thm}
\begin{proof}
We may assume $\mu=0$ (by reasoning with $x_i-\omega(x_i)$ instead of $x_i$). First we prove the statement for $k=2$. Then $\omega(|\frac{x_1+\dotsb+x_n}{n}|)^2=\sum_i\omega(x_i^2)/n^2=\omega(x_1^2)/n$. The first equality follows from independence ($\omega(x_ix_j)=\omega(x_i)\omega(x_j)=0\text{ for }i\neq j$) the second from the fact that they are identically distributed. The case $k=2$ is now trivial. Now let $k=2m$. Then $|x_1+\dotsb +x_n|^k=(x_1+\dotsb +x_n)^k$. Put $s_n=(x_1+\dotsb +x_n)/n$. Expanding $s_n^k$ in a multinomial series we note that  independence and the fact that $\omega(x_i)=0$ implies that all terms in which at least one of the $x_i$ has power 1 do not contribute to $\omega(s_n^k)$. The total number of the remaining terms is $O(n^{m})$. Since the denominator is $n^{2m}$ we see that $\omega(s_n^k)\rightarrow 0$. Since for any $x\in A$, $|x|=(x^2)^{1/2}$ can be approximated by polynomials in $x^2$ we conclude that $\omega(|s_n|)\rightarrow 0$. Finally, using the Cauchy-Schwartz type inequality $\omega(|s_n|^{2r+1})\leq \omega(s_n^2)\omega(s_n^{2r})$ we see that the theorem is true for all $k$. 
\end{proof}
\begin{cor_t} \label{cor:prob_bounds}
Let $x_1,\dotsc, x_n \text{ and } \mu$ be as in the Theorem and set $s_n=(x_1+\dotsb+x_n)/n$. Then for any $\epsilon >0$ there exist $n_0$ such that for all $n>n_0$ 
\[ P(|s_n-\mu|>\epsilon) <\epsilon \]
\end{cor_t}
\begin{proof}
Using Chebysev inequality we have \( P(|s_n-\mu|>\epsilon)= P(|s_n-\omega(s_n)|>\epsilon) \leq \frac{\omega(|s_n-\mu|^2)}{\epsilon^2}\). As $\omega(|s_n-\mu|^2)\rightarrow 0$ (Theorem \ref{thm:weak-law}) there is $n_0$ such that $\omega(|s_n-\mu|^2)<\epsilon^3$ for $n>n_0$. 
\end{proof}
\section{Communication and Information}\label{sec:info}
We now come to our original theme: an algebraic framework for communication and information processes. We can view information as a measure of our state of ignorance or uncertainty. Mathematically, it is equivalent to some measure associated with a probability distribution of some physical quantity which we identify  with an observable. Thus any manipulation of the quantity, for example, transmitting it or measuring it is given by some operation on the observable. Since our primary goal is the modeling of information processes we refer to the simple model of communication in the Introduction and model different aspects of it. 
\subsection{Source and coding}
\begin{Def}
A source  is a pair $\scrp{S}=(X,S)$ where $X\subset A$, $A$ a $\cstar$ algebra and $S$ is a set of states. A source is static if $S$ consists of single state. It is discrete if $X$ is countable. 
\end{Def}
This definition abstracts the essential properties of a source. A real source could be an animate (human speech, for example) or inanimate object (a radio transmitter, for example). Its output can be considered discrete, for example, a keyboard with a fixed alphabet or continuous like radiation from a star. In this work we will be mainly concerned with discrete sources. Then $X$ will be called the {\em source alphabet}. We assume that at each instant there is a probability distribution on the letters of the alphabet characterizing the {\em state} of the source at that instant. Thus a discrete source is a countable set of random variables. In the algebraic view it is a sequence of elements $X$ of a $\cstar$ algebra. The set of states $S$, called the states of the source, provide the probability distributions. If this distribution does not change (equivalently $S$ consists of a single element)  then we have a static source. We will mostly deal with static sources in this work. When we model transmission of information as a Markov process the state of the source is identified with the initial probability distribution. There is dual view. Suppose that a source $\mathscr{S}$ emits letters from a finite alphabet. Then the set $X$ in the above definition is a subset of the atomic basis (corresponding to the alphabet) of the algebra $A$. For a state $\omega$ define 
\[\cali{O}_{\omega}=\sum_{i=1}^n \omega(x_i)x_i, \; \{x_1, \dotsc, x_n\}\text{ an atomic basis }\]
We say that $\cali{O}_{\omega}$ is the output of the source in state $\omega$. Intuitively, $\cali{O}_{\omega}$ is a kind of mean ``point'' in the space of outputs (compare it with the notion of center of mass in mechanics). More importantly, it facilitates calculation of important quantities and has close analogy with the quantum case. The quantum analogue may be pictured as follows. The source outputs ``particles'' in definite ``states'' $x_i$ with probability $p_i=\omega(x_i)$. Note that here state corresponds to a projection operator. A measurement for $x_i$ means applying the dual operator $\omega_i\;(\omega_i(x_j)=\delta_{ij})$ giving $\omega_i(\cali{O}_\omega)= p_i$. 

Let $\mathscr{Z}=(X,\omega)$ be a static discrete source. Suppose every $x\in X$ belongs to a finite-dimensional subalgebra generated by a (finite) set of $\omega$-independent elements. Then using the Theorem \ref{thm:structIndep} we may assume that $A=\inftens B$ where $B$ is finite-dimensional abelian $\cstar$ algebra and $\omega$ is an (infinite) product state. In this case, each element of $X$ is a tensor product of elements of an atomic basis of $B$. In the rest of the paper we assume that $X$ is the product basis of atomic elements. For example, if $B$ is the two dimensional algebra with atomic basis $\{y_0,y_1\}$ then $X$ is the set of elements of the form 
$z_1\tensor z_2\tensor\dotsb \tensor z_k\tensor\unit\tensor\unit\tensor\dotsb$ where $z_i\in\{y_0,y_1\}$. 
\subsection{Source coding}
Let $B$ be a finite-dimensional $\cstar$ algebra and $A=\inftens B$ . We consider $\tensor^n B$ as a subalgebra of $A$ via the standard embedding (all ``factors'' beyond the $n$th place equal $\unit$). Let $X_n$ be its atomic basis in some fixed ordering and let $X=\bigcup_n X_n$. We can consider $B$ as the source alphabet and $X_n$ as strings of length $n$. Let $B'$ be another finite-dimensional $\cstar$ algebra and $A'=\inftens B'$. A source coding is a linear map $f:B\rightarrow T= \subset \sum_{k\geq 1}^m \tensor^k B' $. Here $T$ is the linear subspace. It induces a (linear) map 
\[\tensor^n f:\tensor^n B\rightarrow A'\text{ given by } \tensor^n f(x_1\tensor\dotsb\tensor x_n)=f(x_1)\tensor\dotsb\tensor f(x_n) \] 
$\tensor^nf$ extends to a unique map $F: A\rightarrow A'$. Note that we first induce a map on $\tensor^n B$, $n=1,2,\dotsc, $ and {\em then} lift it to $A$. We allow the map $f$ to take values that are not simple products. However, for classical communication we require that each atomic basis element $x_i\in B$ be mapped to a tensor product of atomic basis elements. Since we are dealing with classical information in this paper it will be implicitly assumed that all the codes are classical. Let us consider an example to clarify these points. 

\vspace{.25cm}
\noindent
{\bf Example.} Let $\{x_0,x_1,x_2,x_3\}$ be an atomic basis for $B$. Let $B'=G$ with atomic basis $\{y_0,y_1\}$. Define $f_1$ by $f_1(x_0)=y_0,f_1(x_1)= y_1, f_1(x_2)=y_0\tensor y_1 \text{ and } f_1(x_3)=y_1\tensor y_0$. Denote by $\hat{f}_1$ its extension to tensor products. Since $\hat{f}_1(x_0\tensor x_1)=y_0\tensor y_1=\hat{f}_1(x_2)$, $\hat{f}_1$ is not injective. Hence it cannot be inverted on its range. Consider next the map $f_2(x_0)=y_0,f_2(x_1)=y_0\tensor y_1,f_2(x_2)=y_0\tensor y_1\tensor y_1\text{ and }f_2(x_3)=y_1\tensor y_1\tensor y_1$. This map is invertible but one has to look at the complete product before finding the inverse. It is not {\em prefix-free}. 

\vspace{.25cm}
\noindent
Now going back to the general formulation a code $f:B\rightarrow T$ is defined to be prefix-free if for distinct members  $x_1,x_2$ in an atomic basis of $B$, $f'(x_1)f'(x_2)=0$ where $f'$ is the map $f': B\rightarrow \inftens B'$ induced by $f$. That is, distinct elements of the atomic basis of $B$ are mapped to {\em orthogonal} elements. Recall that two elements $x,y$ of an algebra are considered orthogonal if their product $xy=0$.  \footnote{The use of the term ``orthogonal'' may be questionable since there is no scalar product. But let us observe that the projection operators corresponding to two pure states in quantum mechanics have algebraic product 0 if and only if they are orthogonal.} Now, in the standard formulation an alphabet is a finite set and a code is a map from $Y\rightarrow Z^+$ where $Y,Z$ are alphabets and $Z^+$ is the set of non-empty finite strings from $Z$. The definition of prefix-free in this case is clear. In the algebraic language the free monoidal structure defined by concatenation is replaced by the tensor structure. Then the
``code-word'' $ z_1\tensor z_1\tensor\dotsb \tensor z_k \tensor \unit\tensor \unit\tensor\dotsb$ is not orthogonal to another $ z'_1\tensor z'_1\tensor\dotsb \tensor z'_m \tensor \unit\tensor \unit\tensor\dotsb$ with $k\leq m$ if and only if $z_1=z'_1,\dotsc, z_k=z'_k$. We observe that one has to be careful about correspondence between the two approaches. For example, one might be tempted to identify the identity $\unit$ with the empty string but the $\unit$ is the sum of the members of an atomic basis! The binary operation ``+'' has a relatively lesser role in the classical formalism but it is crucial in the quantum framework (via superposition principle). Our first result is a useful and well-known inequality proved using algebraic techniques. 

\begin{lem}[Kraft inequality]\label{lem:Kraft}
Let $B$ be an $n$-dimensional abelian $\cstar$ algebra. Corresponding to a finite sequence $k_1\leq k_2\leq \dotsb \leq k_m$ of positive integers let $\alpha_1,\dotsc, \alpha_m$ be a set of prefix-free elements in $\sum_{i\geq 1} \tensor^i B$ such that $\alpha_i\in \tensor^{k_i} B$. Further, suppose that each $\alpha_i$ is a tensor product of elements from a fixed atomic basis of $B$. Then 
\beq \label{eq:kraft}
 \sum_{i=1}^m n^{k_m-k_i} \leq n^{k_m}
\eeq
\end{lem}
\begin{proof}
Let $\mathbbm{b}=\{y_1,\dots,y_n\}$ be the fixed atomic basis of $B$ and set $k_m=M$. We can then restrict our attention to the finite-dimensional algebra $Z=\sum_{i=1}^M \tensor^i B$. Let $\alpha_1=z^1_1\tensor\dotsb \tensor z^1_{k_1}\tensor \unit\tensor \dotsb \tensor\unit$ where $z^1_i\in \mathbbm{b}$. Let $\beta=z^1_1\tensor\dotsb \tensor z^1_{k_1}$ and 
\[ Z_1= \{\beta\tensor \gamma : \gamma \in \tensor^{M-k_1} B\} \]
Then $Z_1\subset \tensor^M B $ is a subalgebra (without unit) of dimension $n^{M-k_1}$. The assumption that $\alpha_i$ are prefix-free implies $\alpha_2,\alpha_3,\dotsc,\alpha_M$ must be in $Z_1'$ the ``orthogonal'' complement to $Z_1$ in $Z$. Dimension of $Z'_1=n^M-n^{M-k_1}$. Repeating this argument with $\alpha_2,\dotsc,\alpha_{k_{m-1}}$ we conclude that $\alpha_{k_m}$ must be in a subspace of dimension $n^M-n^{M-k_1}-n^{M-k_2}-\dotsb -n^{M-k_{m-1}}$. Since $\alpha_{k_m}$ is non-zero $n^M-n^{M-k_1}-n^{M-k_2}-\dotsb -n^{M-k_{m-1}}\geq 1$. This is equivalent to the relation (\ref{eq:kraft}). 
\end{proof}
With the notation of the lemma we call the sequence $W=\{\alpha_1,\dotsc, \alpha_m\}$ decipherable if the tensor product of any two distinct finite ordered sequence of elements from $W$ are distinct. The sequences may have repeated elements. The Kraft inequality is valid for decipherable sequences \cite{McMillan}. However, the proof is essentially combinatorial. The Kraft inequality also provides a sufficiency condition for prefix-free code \cite{Ash,CoverT}. Thus the existence of a decipherable code of word-lengths $(k_1,k_2,\dotsc,k_m)$ implies the existence of a prefix-free code of same word-lengths. In the following, we restrict ourselves to prefix-free codes. If $g:A \rightarrow \tensor^{\infty}B$ is a prefix-free code then it maps orthogonal elements to orthogonal elements. It is therefore an algebra isomorphism (a one-to-one homomorphism). Next we have a technical lemma that is useful in finding bounds. 
\begin{lem}
Let $f$ be a continuous real function on $(0,\infty )$ such that $xf(x)$ is convex and $\lim_{x\rightarrow 0} xf(x)=0$. Let $A$ be a finite-dimensional $\cstar$ algebra with atomic basis $\{x_1,\dotsc, x_n\}$ and $\omega$ a state on $A$. Then for any set of numbers $\{a_i\; :i=1,\dotsc, n;\; a_i>0\text{ and } \sum_i a_i\leq 1\}$ we have 
\[ \omega(\sum_i f(\frac{\omega(x_i)}{a_i})x_i) \geq f(1)\]
\end{lem}
\begin{proof}
Let $\omega(x_i)=p_i$. We have to show that $\sum p_if(p_i/a_i) \geq f(1)$. First assume that all $p_i>0$ and $\sum_i a_i=1$. Then 
\[ \sum_i p_i f(p_i/a_i) =\sum_i a_i\frac{p_i }{a_i}f(\frac{p_i }{a_i}) \geq f(\sum p_i)=f(1)\]
by convexity of $xf(x)$. The general case can be proved by starting with $a_i$ corresponding to $p_i>0$ and adding extra $a_j$'s to satisfy $\sum_i a_i=1$ if necessary. The corresponding $p_j$ is set to $0$. Now define a new function $g(x)=xf(x), \; x>0$ and $g(0)=0$. The conclusion of the lemma follows by arguing as above with $g$.   
\end{proof}
Using the lemma for the function $f(x)=\log {x}$ and Lemma \ref{lem:Kraft} we easily deduce the following. 
\begin{propn}[Noiseless coding]
Let $\mathscr{S}$ be a source with output $\cali{O}_{\omega}\in A$, a finite-dimensional $\cstar$ algebra with atomic basis $\{x_1,\dotsc, x_n\}$ (the alphabet). Let $g$ be prefix-free code such that $g(x_i)$ is a tensor product of $k_i$ members of the code basis. Then 
\[\omega(\sum_i k_ix_i+\log{\cali{O}_{\omega}})\geq 0\]
\end{propn}

Next we give a simple application of Theorem \ref{thm:weak-law}. First define a positive functional $\tr$ on a finite dimensional abelian $\cstar$ algebra $A$ with an atomic basis $\{x_1,\dotsc, x_d\}$ by 
\(\tr =\omega_1+\dotsb+\omega_d\) where $\omega_i$ are the dual functionals. It is clear that $\tr$ is independent of the choice of atomic basis. Informally, the function $\tr$ gives the dimension of a projection.  
\begin{thm}[Asymptotic Equipartition Property (AEP)]\label{thm:AEP}
Let $\scrp{S}$ be a source with output $\cali{O}_{\omega}=\sum_{i=1}^d \omega(x_i)x_i$ where $\omega$ is a state on the finite dimensional algebra with atomic basis $\{x_i\}$. Then given $\epsilon >0$ there is a positive integer $n_0$ such that for all $n>n_0$ 
\[ P(2^{n(H(\omega)-\epsilon)}\leq \tensor^n \cali{O}_{\omega}\leq 2^{n(H(\omega)+\epsilon)}) > 1- \epsilon \]
where $H=\omega(\log_2(\cali{O}_{\omega}))$ is the {\em entropy} of the source and the probability distribution is calculated with respect to the state $\Omega_n=\omega\tensor\dotsm\tensor\omega$ ($n$ factors) of $\tensor^n A$. If $Q$ denotes the identity in the subalgebra generated by $(\epsilon I-|\log_2(\tensor^n\cali{O}_\omega)+nH|)_+$ then 
\[(1-\epsilon)2^{n(H(\omega)-\epsilon)}\leq \tr(Q) \leq 2^{n(H(\omega)+\epsilon)} \]
\end{thm}
Before proving the theorem some explanations are necessary. First $\log_2{x}\;(=\ln{x}/\ln{2})$ is usually defined for strictly positive elements of a $\cstar$ algebra\footnote{Henceforth $\log$ will be always with respect to base 2 unless specified otherwise}. We extend the definition to all non-zero $x\geq 0$. The standard method of extending complex functions (continuous or analytic) functions to a $\cstar$ algebra is via functional calculus \cite{KR1}. However, in our case it is simpler. Let $\{y_i\}$ be a  atomic basis in an abelian $\cstar$ algebra. 
Let $y=\sum_i a_iy_i$ with $a_i\geq 0$. Then define $\log_2{y}=\sum_i b_i y_i$ where $b_i=\log{a_i}$ if $a_i>0$ and 0 otherwise. This definition implies that some standard properties of $\log$ are no longer true (e.g.\ $2^{\log{x}}\neq x$). But in the present context it gives the correct result when we take expectation values as in the formulas in the theorem. A somewhat longer but mathematically better justified route is to ``renormalize'' the state. Thus if $\omega(x_i)=0$ for $k$ indices we define $\omega'(x_i)=\delta$ where $\delta$ is arbitrarily small but positive and $\omega'(x_j)=\omega(x_j)-k\delta$ where $\omega'(x_j)>k\delta$. If we can prove the theorem now for $\omega'$ and since the relations are valid in the limit $\delta\rightarrow 0$ then we are done. We will not take this path but implicitly assume that the probabilities are positive. Finally, note that the element $Q$ is a projection on the subalgebra generated by $(\epsilon I-|\log_2(\tensor^n\cali{O}_\omega)-nH|)_+$. It corresponds to the set of strings whose probabilities are between $2^{-nH-\epsilon}$ and $2^{-nH+\epsilon}$. The integer $\tr(Q)$ is simply the cardinality of this set. 
\begin{proof}[Proof of the theorem]
First note that $\log{ab}=\log a+\log b$ for elements $a,b\geq 0$ in $A$. We can write $\tensor^n\cali{O}_{\omega}=X_1X_2\dotsb X_n$ where $X_i=\unit\tensor\unit\tensor\dotsm\tensor\cali{O}_{\omega}\tensor\unit\tensor\dotsm\tensor\unit$ with $\log{\cali{O}_{\omega}}$ in the $i$th place. The fact that $\Omega_n$ is a product state on $\tensor^n A$ (corresponding to a source whose successive outputs are independent) implies that $X_i$ are independent and identically distributed. We can now apply the corollary to Theorem \ref{thm:weak-law} yielding 
\(P(|\log{(\tensor^n\cali{O}_{\omega})} -\Omega_n(\log{X_1})|>\epsilon)=P(|\log{(\tensor^n\cali{O}_{\omega})} -\omega(\log{(\cali{O}_{\omega})})|>\epsilon)\). 
\end{proof}

\subsection{Communication Channels}
Every form of communication requires channels through which signals are sent and received. It is perhaps the most important component in the mathematical models of communication. We will not be dealing with real channels which are complex physical objects--- the atmosphere, a telephone cable, a bus on the mainboard of a computer are some examples. Our object is to give simple mathematical models of a channel which still yield interesting results relevant for concrete channels. The original paper of Shannon characterized channels by a transition probability function. Thus, the channel (precisely a two-way channel) has an input alphabet $X$ and output alphabet $Y$ and a sequence of random functions $\phi_n: X^n\rightarrow Y^n$. The latter are characterized by probability distributions $p_n(y^{(n)}|x^{(n)})$, the interpretation being: $\phi_n(x^{(n)})=y^{(n)}$ with conditional probability  $p_n(y^{(n)}|x^{(n)})$. Note that the distribution depends on the entire history. We say that such a channel has (infinite) memory. A channel has finite memory if there is an integer $k\geq 0$ such that if $x^{(n)}=x_nx_{n-1}\cdots x_{n-k+1}\dotsc x_1$ then $p_n(y^{(n)}|x^{(n)})= p_n(y^{(n)}|x'^{(n)})$ for any string $x_n'$ of length $n$ such that $x_n'=x_n, \dotsc, x_{n-k+1}'=x_{n-k+1}$. That is, the probability distribution depends on the most recent $k$ symbols seen by the channel. A channel is {\em memoryless} if $k=1$. Since we will be dealing mostly with discrete memoryless channels (DMS) this property will be tacitly assumed unless stated otherwise. In the memoryless case it is easy to show the simple form of transition probabilities 
\beq
p_n(y^{(n)}|x^{(n)})=p_n(y_1\dotsc y_n|x_1\dotsc x_n)=p(y_1|x_1)p(y_2|x_2)\dotsb p(y_n|x_n) 
\eeq
This motivates us to define the {\em channel transformation matrix} $C(y_j|x_i)$ with $y_j\in Y$ and $x_i\in X$. As before in this work $X$ and $Y$ will be finite sets. Since the matrix $C(y_j|x_i)$ is supposed to represent the probability that the channel outputs $y_j$ on input $x_i$ we must have $\sum_j C(y_j|x_i)=1$ for all $i$. In other words, matrix $C(ij)=C(y_j|x_i)$ is {\em row stochastic}. This is the standard formulation.  \cite{Ash,CoverT,Khinchin}\footnote{In this work we will not deal with channel coding and decoding. Including these concepts is not difficult but complicates the notation.} We now turn to the algebraic formulation. We restrict ourselves to two-terminal channels here. 
\begin{Def}
A DMS channel $\cali{C}=\{X,Y,C\}$ where $X$ and $Y$ are abelian $\cstar$ algebras of dimension $m$ and $n$ respectively and $C: Y \rightarrow X$ is a unital positive map. The algebras $X$ and $Y$ will be called the input and output algebras of the channel respectively. Given a state $\omega$ on $X$ we say that $(X,\omega)$ is the input source for the channel. 
\end{Def}
We recall that a positive map $C:Y\rightarrow X$ is a linear map such that $C(y)\geq 0$ if $y\geq 0$. Sometimes we write the entries of $C$ in the more suggestive form $C_{ij}=C(y_j|x_i)$ where $\{y_j\}$ and $\{x_i\}$ are atomic bases for $Y$ and $X$ respectively. Thus $C(y_j)=\sum_i C_{ij}x_i= \sum_i C(y_j|x_i)x_i$. Note that in our notation $C$ is an $m\times n$ matrix. Its transpose $C^T_{ji}=C(y_j|x_i)$ is the channel matrix in the standard formulation. We have to deal with the transpose because the channel is a map {\em from} the output alphabet to the input alphabet. This may be counterintuitive but observe that any map $Y\rightarrow X$ defines a unique dual map $\cali{S}(X)\rightarrow \cali{S}(Y)$, on the respective state spaces. Informally, a channel transforms a probability distribution on the input alphabet to a distribution on the output. In other words, given an input source there is a unique output source determined by the channel. Let us note that in case of abelian algebras every positive map is guaranteed to be {\em completely positive} \cite{Tak1}. This is no longer true in the non-abelian case. Hence for the quantum case completely positivity has to be explicitly imposed on (quantum) channels. 

We characterize a channel by input/output algebras (of observables) and a positive map. Like the source output we now define a useful quantity called {\em channel output}. Corresponding to the atomic basis $\{y_i\}$ of $Y$ let $\tensor^k y_{i(k)}$ be an atomic basis in $\tensor^n Y$. Here $i(k)=(i_1i_2\dotsc i_k)$ is a multi-index. Similarly we have an atomic basis $\{\tensor^k x_{j(k)}\}$ for $\tensor^k X$. The level-$k$ channel output is defined to be. 
\beq\label{eq:chOutput}
O^k_C = \sum_{i(k)} y_{i(k)}\tensor C^{(k)}(y_{i(k)})  
\eeq
Here $C^{(k)}$ represents the channel transition probability matrix on the $k$-fold tensor product corresponding to strings of length $k$. In the DMS case it is simply the $k$-fold tensor product of the matrix $C$. The channel output defined here encodes most important features of the communication process. First, given the input source function \footnote{We called this the source output before. But as the channel has two terminals we call it  input source function to avoid confusion.} \(\cali{I}_{\omega^k}=\sum_i\omega^k(x_{i(k)})x_{i(k)}\) the output source function is defined by 
\beq \label{eq:chOutput2}
\cali{O}_{\tilde{\omega}^k} = I\tensor \tr_{\tensor^kX}((\unit \tensor \cali{I}_{\omega^k})O^k_c)=\sum_i \sum_j C(y_{i(k)}|x_{j(k)})\omega^k(x_{j(k)})y_{i(k)}
\eeq
Here, the state $\tilde{\omega}^k$ on the output space $\tensor^k Y$ can be obtained via the dual $\tilde{\omega}^k(y)=\tilde{C}^k(\omega^k)(y)=\omega^k(C^k(y))$. The formula above is an alternative representation which is very similar to the quantum case. The {\em joint output} of the channel can be considered as the combined output of the two terminals of the channel. This is obtained by {\em not} tracing out over the input in the equation \ref{eq:chOutput2}. Thus the joint output
\beq \label{eq:chOutputJoint}
\begin{split}
&\cali{J}_{\tilde{\Omega}^k} = (\unit \tensor \cali{I}_{\omega^k})O^k_C=\sum_{ij} \Omega^k(y_{i(k)}\tensor x_{j(k)})y_{i(k)}\tensor x_{j(k)} \text{ with } \\
& \Omega^k(y_{i(k)}\tensor x_{j(k)})= C(y_{i(k)}|x_{j(k)})\omega(x_{j(k)})
\end{split}
\eeq
Let us analyze the algebraic definition of channel given above. For simplicity of notation, we restrict ourselves to level 1. The explicit representation of channel output is 
\[\sum_i y_i\tensor \sum_j C(y_i|x_j)x_j \]
We interpreted this as follows: if on the channel out-terminal $y_i$ is observed then the input could be $x_j$ with probability \(C(y_i|x_j)\omega(x_j)/\sum_jC(y_i|x_j)\omega(x_j)\). Now suppose that for a fixed $i$  $C(y_i|x_j)=0$ for all $j$ except one say, $j_i$. Then on observing $y_i$ at the output we are certain that the the input is $x_{j_i}$. If this is true for all values of $y$ then we have an instance of a lossless channel. It is easy to write the channel matrix  in this case. Thus, given $1\leq j\leq n$ let $d_j$ be the set of integers $i$ for which $C(y_i|x_j)> 0$. The lossless property implies that $\{d_j\}$ form a partition of the set $\{1,\dotsc, m\}$. The corresponding channel output is 
\[ O_C= \sum_j \Bigl(\sum_{i\in d_j} C(y_i|x_j)y_i\Bigr)\tensor x_j\]
Clearly lossless channels are the most useful for communication of information. At the other extreme is the {\em useless} channel in which there is no correlation between the input and the output. 
To define it formally, consider a channel $\cali{C}=\{X,Y,C\}$ as above. The map $C$ induces a map $C': Y\tensor X\rightarrow X$ defined by $C'(y\tensor x)=xC(y)$. Given a state $\omega$ on $X$ the 
dual of the map $C'$ defines a state $\Omega_C$ on $Y\tensor X$: \(\Omega_C(y\tensor x)=\omega(C'(y\tensor x))=C(y|x)\omega(x)\). We call $\Omega_C$ the joint (input-output) state of the channel. A channel is 
useless if $Y$ and $X$ (identified as $Y\tensor \unit$ and $\unit \tensor X$ resp.) are $\Omega_C$-independent. 
\begin{lem}
A channel $\cali{C}=\{X,Y,C\}$ with input source $(X,\omega)$ is useless iff the matrix $C_{ij}=C(y_j|x_i)$ is of rank 1. 
\end{lem}
\begin{proof}
Suppose $\cali{C}$ is useless. Note that $\Omega_C(\unit\tensor x)=\omega(x)$ and $\Omega_C(y \tensor \unit)=\tilde{\omega}(y)$ where $\tilde{\omega}(y)=\omega(C(y))$ is the image of 
$\omega$ under the dual of the map $C$. Then $\Omega_C$ independence implies $C(y_j|x_i)\omega(x_i)=\omega(x_i)\tilde{\omega}(y_j)$. We may assume that all $\omega(x_i)>0$ (otherwise we just discard it). 
Hence, $C(y_j|x_i)=\tilde{\omega}(y_j)$ and this proves necessity. Now if $C_{ij}$ has rank 1 then all the rows are non-zero multiples of any one row, say the first. Since $C$ is a row stochastic matrix 
the rows must be identical, that is, $C_{ij}=a_j=\tilde{\omega}(y_j)$ and independence is trivially verified. 
\end{proof}
The definition of a useless channel captures the intuition that if there is no correlation between the input and output then we can recover practically nothing. 
The {\em channel coding} theorem asserts that apart from this extreme case we can decode the output to recover a large portion of the input with high probability of success.  The algebraic version of the channel coding theorem assures that it is possible to approximate, in the long run, an arbitrary channel (excepting the useless case) by a lossless one. 

\begin{thm}[Channel coding]\label{thm:ch_coding}
Let $\cali{C}$ be a channel with input algebra $X$ and output algebra $Y$. Let \(\{x_i\}_{i=1}^n\text{ and } \{y_j\}_{j=1}^m\) be atomic bases for $X$ and $Y$ resp. Given a state $\omega$ on $X$, if the channel is not useless then  for each $k$ there are subalgebras \(Y_k\subset \tensor^k Y, X_k\subset \tensor^k X\), a map $C_k: Y_k\rightarrow X_k$ induced by $C$ and  a lossless channel $L_k: Y_k\rightarrow X_k$ such that 
\[\lim_{k\rightarrow \infty} \Omega(|O_{C_k}-O_{L_k}|) = 0 \text{ on } T_k=Y_k\tensor X_k\]
Here $\Omega=\tensor^{\infty}\Omega_C$ and on $\tensor^k Y\tensor \tensor^k Y$ it acts as $\Omega^k=\tensor^k\Omega_C$ where $\Omega_C$ is the state induced by the channel and a given input state $\omega$. Moreover, if $r_k=\text{dim}(X_k)$ then $R=\frac{\log{r_k}}{k}$, called transmission rate, is independent of $k$. 
\end{thm}
First let us clarify the meaning of the above statements. The theorem simply states that on the chosen set of codewords the channel output of $C_k$ induced by the given channel can be made arbitrarily close to that of a lossless channel $L_k$. Since a lossless channel has a definite decision scheme for decoding the choice of $L_k$ is effectively a decision scheme for decoding the original channel's output when the input is restricted to our ``code-book''. This in turn implies that the probability of error tends to 0. 
\begin{proof}
From an atomic basis of $\tensor^k X$ choose a subset $A_k$ of cardinality $r_k$ (to be determined). Let $X_k$ be the subalgebra generated by $A_k$. Write $C^{(k)}$ for the $k$-fold tensor product of $C$. Let $Q_k$ be the identity on $X_k$ (it is the sum of all the members of $A_k$). For an atomic basis $B_k$ of $\tensor^k Y$ let $B'_k$ be the subset such that $C^{(k)}(y)Q_k\neq 0\text{ for } y\in B'_k$. Let $Y_k$ be the subalgebra generated 
%%%changed P to Q_k here
by $B'_k$ and $C_k: Y_k\rightarrow X_k$ denote the linear map $C_k(y)=Q_kC^{(k)}(y)$. Informally, if we restrict the messages to observables in $A_k$ then the output algebra is $Y_k$. The new channel map is $C_k$. We now have a new channel $\tilde{\cali{C}}^k=(X_k,Y_k,C_k)$. Throughout the rest of the proof we will assume that we are working in $T_k$ with the appropriate maps. We next define $L_k$ as follows. For $y_i\in B'_k$ let \(C_k(y_i)=\sum_j C_k(y_i|x_j)x_j,\:x_j\in A_k\). Let $C_k(y_i|x_{i_r})$ be the maximum of $C_k(y_i|x_j)$ for fixed $y_i$ (if there are more than one index equal to this maximum choose one arbitrarily). Let $L_k(y_i)=\tilde{\omega}(y_i) x_{i_r}$. The map $L_k$ is {\em not} unital. Strictly speaking $L_k$ is not a channel map as we have defined above. However, as we see below, $L_k$ does approximate $O_{C_k}$ in $T_k$ 
with small error. What this means is that with high probability we can correctly associate a unique and correct input to a given channel output \footnote{We have combined two types of decoding scheme: the ideal observer decoding \cite{Ash} and typical set decoding \cite{CoverT}}. The non-unital property of $L_k$ is  reflective of the situation in which some of the original messages {\em outside} of $X_k$ may end up in $Y_k$. Set 
$r_k=2^{kR}$ and let 
\[O_{\tilde{\omega}^k}=\sum_{y\in B_k'}\tilde{\omega}(y)(y\tensor \unit) \text{ and } O_{\omega^k}=\sum_{x\in A_k} \omega(x)(\unit\tensor x)\]
Here $O_{\tilde{\omega}^k}$ and $O_{\omega^k}$ are respectively the input and output source function for the channel  $\tilde{\cali{C}}^k$.  
Let $Z_k$ be the identity on the ideal generated by \((\log{O_{C_k}}- \log{O_{\tilde{\omega}^k}}-k(R+\epsilon))_+=(\log{(O_{C_k}O_{\tilde{\omega}^k}^{-1})}-k(R+\epsilon))_+, \: \epsilon > 0\) in $T_k$. \footnote{This ideal is $T_k(\log{(O_{C_k}O_{\tilde{\omega}^k}^{-1})}-k(R+\epsilon))_+$. Note that we write the scalar $k(R+\epsilon)$ instead of the more accurate $k(R+\epsilon)\unit_k$ where $\unit_k$ is the unit in $T_k$.}. Note that  $O_{C_k}=O_{\Omega^k}O_{\omega^k}^{-1}$ on $T_k$. Since 
\[
Z_k|O_{C_k}O_{\tilde{\omega}^k}^{-1} -2^{k(R+\epsilon)}|=(O_{C_k}O^{-1}_{\tilde{\omega}^k} -2^{k(R+\epsilon)})_+=Z_k(O_{\Omega^k}O_{\omega^k}^{-1}O^{-1}_{\tilde{\omega}^k} -2^{k(R+\epsilon)})\geq 0
\]
and $Z_k^2=Z_k$ we conclude that \(Z_kO_{\omega^k}\leq Z_kO_{\Omega^k}O_{\tilde{\omega}^k}^{-1}2^{-k(R+\epsilon)}\leq Z_k2^{-k(R+\epsilon)}\). The last inequality follows from the fact that $O_{\Omega^k}O_{\tilde{\omega}^k}^{-1}\leq \unit$. We also have $O_{L_k}\leq O_{C_k}$ and $\Omega^k(Z_k)=\tr(Z_kO_{\Omega^k})$. The last fact is true for any projection as can be verified using an atomic basis. We now have 
\[
\begin{split}
&\Omega(Z_k|O_{C_k}-O_{L_k}|)=\Omega^k(Z_k(O_{C_k}-O_{L_k})) \leq \Omega^k(Z_k) =\tr(Z_kO_{\Omega^k})\\
&\leq \tr(Z_kO_{\omega})\leq 2^{-k(R+\epsilon)}\tr(Z_k)\leq 2^{-k(R+\epsilon)}r_k=2^{-k\epsilon}
\end{split}
\]
Hence \(\Omega(Z_k|O_{C_k}-O_{L_k}|)=\Omega^k(Z_k(O_{C_k}-O_{L_k}))\rightarrow 0\) as $k\rightarrow \infty$. To complete the proof we look at the complementary part: \((\unit_k -Z_k)|O_{C_k}-O_{L_k}|\) where $\unit_k$ is the identity in $T_k$. Consider the projection $\unit_k -Z_k$. $Z_k$ is the identity in the annihilating ideal of $F_{k-}$ where $F_k=(\log{O_{C_k}}- \log{O_{\tilde{\omega}^k}}-k(R+\epsilon))$. Let $G_k=(\log{(\tensor^kO_CO_{\tilde{\omega}}^{-1})}-k(R+\epsilon)\unit)$. Then since $F_k$ is the restriction of $G_k$ to a subspace $G_k=F_k+F'_k$ there is an \(F_k'\in \tensor^k Y\tensor\tensor^kX\) with $F_kF'_k=0$ (we use the fact the channel is memoryless). Hence taking an approximating polynomial sequence $G_{k-}=F_{k-}+F'_{k-}$. It follows that  $F_{k-}\leq G_{k-}$ and $Z'_k$, the identity on the annihilating ideal of $G_{k-}$ satisfies $Z'_k\leq Z_k$.  
This implies $\Omega(\unit_k -Z_k)\leq \Omega(\unit-Z'_k)$. By definition 
\(\Omega(\unit-Z_k) =P(\log{\tensor^kO_CO^{-1}_{\tilde{\omega}}}/k- (R+\epsilon)<0)\)
is the probability that $G_k <R+\epsilon$. But 
\[\Omega(|(\log{\tensor^kO_CO_{\tilde{\omega}}})/k -\Omega(\log{O_C}-\log{O_{\tilde{\omega}}})\unit|)\rightarrow 0 \text{ as }k\rightarrow \infty\]
 follows from the law of large numbers (see Theorem \ref{thm:weak-law} and its corollary). The quantity $I(X,Y)= \Omega(\log{O_C}-\log{O_{\tilde{\omega}}})=H(Y)-H(Y|X)$ is defined as the {\em mutual information} between the input and output algebras and $H(Y|X)$ is the conditional entropy. Thus if we have $R<I(X,Y)$, say $R\leq  I(X,Y)-2\epsilon$ then \(\Omega(\unit-Z'_k)=P(\log{\tensor^kO_CO^{-1}_{\tilde{\omega}^k}}/k- (R+\epsilon)<0)\leq P(|\log{(\tensor^kO_CO^{-1}_{\tilde{\omega}})  }/k- \unit|>\epsilon)\) but the latter $\rightarrow 0$. Putting it all together we have for any $\epsilon >0$ and $R<I-2\epsilon$
 \[
 \begin{split}
 & \Omega((\unit_k -Z_k)|O_{C_k}-O_{L_k}|)\leq \Omega(\unit_k -Z_k)\leq \Omega(\unit -Z'_k)\\
 &= P(|\log{(\tensor^kO_CO^{-1}_{\tilde{\omega}})  }/k- \unit|>\epsilon) \rightarrow 0 \text{ as } k\rightarrow \infty \\
 \end{split}
 \]
 As we already have \( \Omega(Z_k|O_{C_k}-O_{L_k}|)\rightarrow 0 \) the proof is complete. 
\end{proof}
The channel coding theorem implies that it is possible to choose a set of ``codewords'' which can be transmitted with high reliability. It is easy to see that for a lossless channel the input entropy $H(X)$ is equal to the mutual information. We may think of this as conservation of entropy or information which justifies the term ``lossless''. Since it is always the case that $H(X)-H(X|Y)=I(X,Y)$ the quantity $H(X|Y)$ can be considered the loss due to the channel. The channel coding theorem is perhaps the most celebrated theorem in Shannon's work although his proof was not rigorous. The algebraic version of the theorem serves two primary purposes. First, we attempt to make the proof as ``algebraic'' as possible. More importantly, it gives us the commutative perspective from which we will seek possible extensions to the non-commutative case. Secondly, the channel map $L$ can be used for a decoding scheme. Thus we may think of a coding-decoding scheme for a given channel as a sequence of pairs $(X_k, L_k)$ as above. 

The coding theorems can be extended to more complicated scenarios like ergodic sources and channels with finite memory. The converse of the channel coding theorem---roughly, any such coding scheme with error tending to 0 (convergence in probability) must have the rate $\log{r_k}/k \leq I$---is also true. We will not pursue these issues further here. But we are confident that these generalizations can be appropriately formulated and proved in the algebraic framework.  
\section{Conclusion and preview of the future work}
In the preceding sections we have laid the basic algebraic framework for information theory. This work was devoted to classical parts of information theory corresponding to abelian algebras. Since information theory relies heavily on probabilistic concepts we devoted a major part of the paper to algebraic probability theory. Although, we often confined our discussion to finite-dimensional algebras corresponding to finite sample spaces it is possible to extend it to infinite-dimensional algebras of continuous sample spaces. In this regard, a natural question is: can the algebraic formulation replace  Kolmogorov axiomatics based on measure theory? Naively, the answer is no because the assumption of a norm-compete algebra imposes the restriction that the random variables that they represent must be {\em bounded}. Moreover, the GNS construction implies that the algebraic framework is essentially equivalent to (almost) bounded random variables on a locally compact space. In order to deal with the unbounded case we have to go beyond the normed algebra structures. A possible course of action is indicated in the examples given in section \ref{sec:examples}: via the use of a ``cut-off''. A more general approach would be to consider sequences which converge in a topology weaker than the norm topology to elements of a larger algebra. These and other related issues on foundations are deep and merit a separate investigation. 

The second major theme of this paper is information theory in the algebraic framework. As some the most important results of information theory concern finite or discrete alphabet we have primarily dealt with these cases only. In this context, we can treat ergodic sources, channels with finite memory and multi-terminal channels. These topics will be  investigated in the future in the non-commutative setting. However, let us recall one of the principal motivation of this paper: the construction of a single framework for dealing with quantum and classical information. We have seen that the algebraic theory in the commutative case already indicates the close analogies between the two cases. We will delve deeper into these analogies and aim to throw light on some basic issues like quantum Huffman coding \cite{Braunstein}, channel capacities and general no-go theorems among others, once we formulate the appropriate models. In this context, let us mention that many investigators have recognized the importance of the algebraic framework but a comprehensive algebraic model which can be extended to infinite-dimensional case is lacking. We aim to address these important issues in subsequent work.

\end{document}